\documentclass[9pt,twocolumn]{extarticle}

\pdfoutput=1

\usepackage{amsmath, amsthm, amssymb}
\usepackage[ruled]{algorithm2e} 
\usepackage[backend=bibtex,sorting=none, style=science]{biblatex}
\usepackage{complexity}
\usepackage[usenames, dvipsnames]{color}
\usepackage[margin=1in]{geometry}
\usepackage{graphicx}
\usepackage{imakeidx} 
\usepackage{pifont}
\usepackage{siunitx}
\usepackage{soul}
\usepackage{titling}

\DeclareNameAlias{sortname}{first-last}
\DeclareNameAlias{default}{first-last}
\newenvironment{proofof}[1]{\begin{trivlist}\item[]{\flushleft\it
Proof of~#1.}}
{\qed\end{trivlist}}

\bibliography{biblio} 

\renewcommand{\figurename}{Fig.}

\newtheorem{theorem}{Theorem}

\newcommand{\pseudosection}[1]{\vspace{9pt}\noindent\textbf{#1}}

\newcommand{\ket}[1]{\ensuremath{|#1\rangle\mkern-1mu}}
\newcommand{\bra}[1]{\ensuremath{\mkern-1mu\langle#1|}}

\def\({\left(}
\def\){\right)}
\def\[{\left[}
\def\]{\right]}

\def\mytitle{Witnessing eigenstates for quantum simulation of Hamiltonian spectra}

\title{\vspace{-1.0cm}\LARGE\textbf{\textrm{\color{black} \mytitle}}}

\author{R.~Santagati$^{1\dagger*}$, J.~Wang$^{1\dagger}$, A.~A.~Gentile$^{1\dagger}$, S.~Paesani$^{1}$, N.~Wiebe$^{2*}$, J.~R.~McClean$^{3, 4}$,  S.~Morley-Short$^{1, 5}$,
\\P.~J.~Shadbolt$^{6}$, D.~Bonneau$^{1}$, J.~W.~Silverstone$^{1}$, 
D.~P.~Tew$^{7}$, X.~Zhou$^{8}$, J.~L.~O\textquoteright Brien$^{1}$ and
M.~G.~Thompson$^{1*}$}
\date{}
\begin{document}
\twocolumn[{
\maketitle
\vspace{-7mm}
\begin{center}
\begin{minipage}{0.8\textwidth}
\begin{center} 
\textit{\textrm{\small\textsuperscript{1}
Quantum Engineering Technology Labs, H. H. Wills Physics Laboratory and Department of Electrical and Electronic Engineering, University of Bristol, BS8 1FD, UK
\\\textsuperscript{2} Quantum Architectures and Computation Group, Microsoft Research, Redmond, Washington 98052, USA
\\\textsuperscript{3} Computational Research Division, Lawrence Berkeley National Laboratory, Berkeley, CA 94720, USA
\\\textsuperscript{4} Google Inc., Venice, California 90291, USA
\\\textsuperscript{5} Quantum Engineering Centre for Doctoral Training, H. H. Wills Physics Laboratory and Department of Electrical and Electronic Engineering, University of Bristol, Tyndall Avenue, BS8 1FD, United Kingdom
\\\textsuperscript{6} Department of Physics, Imperial College London, London, SW7 2AZ, UK
\\\textsuperscript{7} School of Chemistry, University of Bristol, Bristol BS8 1TS, United Kingdom
\\\textsuperscript{8} State Key Laboratory of Optoelectronic Materials and Technologies and School of Physics, Sun Yat-sen University, Guangzhou 510275, China
\\\textsuperscript{$\dagger$} These authors contributed equally to this work.
\\\textsuperscript{*} e-mails: raffaele.santagati@bristol.ac.uk, nawiebe@microsoft.com, mark.thompson@bristol.ac.uk }} 

\end{center}
\end{minipage}
\vspace{+2mm}
\end{center}

\setlength\parindent{12pt}

\begin{quotation}
\noindent
\textbf{Abstract:}
The efficient calculation of Hamiltonian spectra, a problem often intractable on classical machines, can find application in many fields, from physics to chemistry. 
Here, we introduce the concept of an  "eigenstate witness" and through it provide a new quantum approach which combines variational methods and phase estimation to approximate eigenvalues for both ground and excited states.
This protocol is experimentally verified on a programmable silicon quantum photonic chip, a mass-manufacturable platform, which embeds entangled state generation, arbitrary
controlled-unitary operations, and projective measurements. Both ground and excited states are experimentally found with fidelities $>99\%$, and their eigenvalues are estimated with 32-bits of precision. 
We also investigate and discuss the scalability of the approach and study its performance through numerical simulations of more complex Hamiltonians. 
This result shows promising progress towards quantum chemistry on quantum computers.
\vspace{+6mm}
\end{quotation}

}]		


\section*{Introduction}
The simulation of quantum mechanical systems, using conventional classical methods, requires resources that make the problems rapidly intractable when the size of the system grows~\cite{Feynman1982}. 
Since Feynman's seminal proposal, several algorithms for quantum simulation have followed~\cite{Lloyd1996, AspuruGuzik2005,Georgescu:2014um}, and many demonstrations have been reported in different physical systems \cite{AspuruGuzik2012, Bloch2012,Blatt2012,  Lanyon2011,Brown2006, Neeley2009,Du2010, OMalley2016,Lanyon2010, Peruzzo2014, Ma2011, pitsios2016,Kandala2017}. 
Calculating the spectrum of a given Hamiltonian is a fundamental problem of widespread applicability, necessary, for example, to understand reaction rates or optical spectra in quantum chemistry. In particular, characterisation of excited states is required to study 
energy and charge transfer processes such as those in bulk heterojunction solar cells or photosynthetic light harvesting complexes~\cite{Olsson2006,Lidar1999}. 
These kinds of problems are hard for classical computers
\textcolor{black}{and in the most general case also for quantum computers. However, quantum devices are expected to provide scalable solutions to some instances of interest~\cite{Lloyd1996, AspuruGuzik2005}. Furthermore, also in those cases where classical methods can successfully describe the ground state (e.g. for weakly interacting Hamiltonians), excited states are often hard to access~\cite{serrano2005quantum,Crawford_2000}, 
increasing the interest towards quantum methods able to address the problem of finding an efficient description of excited states. Here we show promising progress in this direction by introducing the concept of an  ``eigenstate witness'', a quantity which has no known efficient analogue in classical algorithms. This witness is based on the entropy acquired by a quantum register, whose time evolution is controlled by an ancillary qubit.} 

Given an approximate eigenstate, the quantum phase estimation algorithm (QPEA) can efficiently estimate the corresponding eigenvalue. A more practical version, the iterative phase estimation algorithm (IPEA)~\cite{Kitaev1995,Dobsicek2007}, has been demonstrated using different quantum hardware, such as nuclear magnetic resonance, photonic and superconducting systems~\cite{Du2010, Lanyon2010, OMalley2016}. 
In order to prepare the input eigenstates, adiabatic state preparation has been proposed as a potentially scalable solution~\cite{AspuruGuzik2005}, at the cost of expensive state preparation and high-depth circuits, making it unsuitable for near-term implementations on quantum computers.

Variational quantum eigensolvers (VQE), using a hybrid quantum-classical approach, were designed to address these shortcomings~\cite{ Peruzzo2014, Jarrod2016, OMalley2016, Romero:2017vp,  Guerreschi:2017vm,Li:2017gm}. \textcolor{black}{Such methods prepare states described via a chosen 
parametrisation, known as \textit{ansatz}, leveraging pre-existing knowledge about the system. Different types of ansatz have been proposed for the variational search, such as unitary coupled-cluster (UCC), which is among the most promising ones to tackle quantum chemistry problems~\cite{Peruzzo2014,Jarrod2016,Wecker2015}.}
In addition, VQE methods are believed to have unique robustness to certain errors in estimating the ground state and its eigenvalue~\cite{Jarrod2016b,OMalley2016}.  
They are, however, quadratically less precise than QPEA, as they rely on sampling for the energy estimation in the original formulation.
Crucially, variational methods so far could only
target ground states. 

A linear response methodology and a spectrum folding method have been proposed as possible solutions~\cite{Jarrod2016b, Peruzzo2014, Colless:2017wb}. However, though the linear response methodology maintains the low coherence time advantages of the original VQE, it requires additional sampling measurements and cannot refine approximate excited states. The folded spectrum method requires instead a quadratic increase in the number of terms of the effective Hamiltonian and consequently in the computational cost of the procedure.  
\textcolor{black}{Thus, experimental demonstrations have been limited to ground states, despite the practical importance of excited states. }

In this work, \textcolor{black}{by introducing the concept of an ``\textit{eigenstate witness}'' we develop a new method that targets also excited states.}
 A crucial limitation for the solution of the eigenvalue problem is that no  method for eigenstates preparation is expected to be scalable in general~\cite{cipra2000ising}. It remains unanswered, whether variational methods can solve particular classes of this problem. However, it is widely conjectured that eigenstates of physically relevant Hamiltonians often can be efficiently represented within an ansatz~\cite{AspuruGuzik2005, Peruzzo2014, Georgescu:2014um, Jarrod2016}. 
In such cases, we estimate that the number of application of controlled operations required to perform our algorithm increases polynomially with the size of the system.

We demonstrate this method in a proof of principle experiment and test its performance via numerical simulations on higher dimensional Hamiltonians. 
For the implementation of the algorithm, we developed a two-qubit quantum photonic processor on the silicon photonic platform~\cite{Bonneau2016,Sharping2006}. This device embeds the key functionalities of on-chip 
entangled states generation ~\cite{Santagati:2017uf, Silverstone2015,Sansoni2010}, tomography~\cite{Santagati:2017uf, Wang2015} and arbitrary controlled-unitary operations (${C\hat{U}}$). 
To perform the latter, an entanglement-based scheme~\cite{Zhou2011,Patel2016} was implemented here for the first time in integrated quantum photonics.

\section*{Results}
\textbf{The WAVES protocol}.  The approach proposed here,  witness-assisted variational eigenspectra solver (WAVES), is divided into three main steps (see Fig.~\ref{figure1}a): 
\textit{i}) an ansatz-based variational search for the ground state,
\textit{ii}) a witness-assisted variational search for excited states, starting with an initial guess obtained from the ground state reference as outlined below, and
\textit{iii}) IPEA for the accurate energy estimate of the eigenstates found.

\begin{figure*}[ht!] \begin{center}
\makebox[\textwidth][c]{\includegraphics[width=1\textwidth]{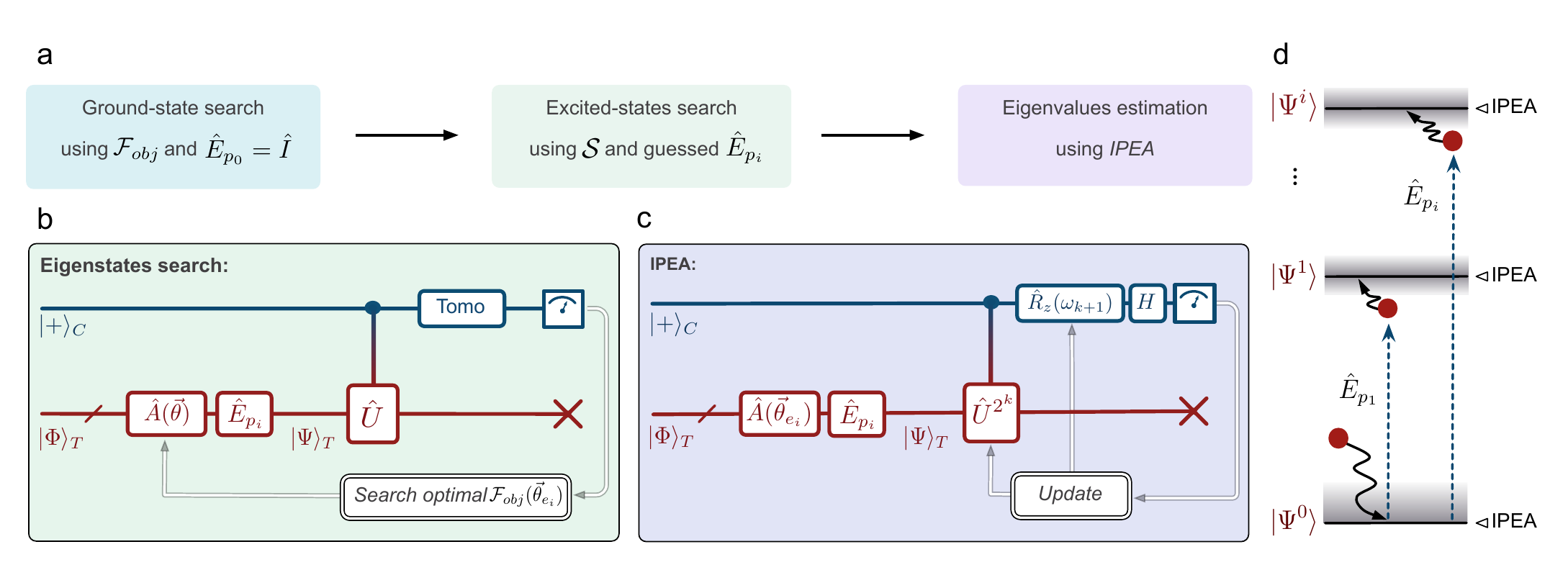}}
 \caption{\textbf{The WAVES protocol. }
(\textbf{a}) 
Flowchart describing the protocol. 
The optimisation of $\mathcal{F}_{\textnormal {obj}}(\vec{\theta})=\mathcal{E}+T\mathcal{S}$ using the circuit in (\textbf{b}) allows one to variationally find the ground state of the Hamiltonian, preparing trial states via the ansatz $\hat{A}(\vec{\theta})$ with no perturbation ($\hat{E}_{p_0} = \hat{I}$). An initial guess for an  excited state is given by a perturbation $\hat{E}_{p_i}$ on the ground state, and then refined using the same circuit, by exploiting the eigenstate witness $\mathcal{F}_{\textnormal {obj}}(\vec{\theta})=\mathcal{S}$ (high-$T$ limit).
(\textbf{c}) For each target eigenstate found, the eigenvalues are precisely estimated via the iterative phase estimation algorithm (IPEA), using the quantum logic circuit, where $H$ is the Hadamard gate. The colour coding in  \textbf{b} and \textbf{c}, blue for the control and red for the target refers to the difference in wavelength between the photon in the control qubit and the one in the target register, in our experimental implementation. 
(\textbf{d})
Diagram representing schematically the intuition behind the proposed approach, where initial guesses of excited states are variationally refined using the witness and IPEA returns the eigenvalues.
  }\label{figure1}
 \end{center}\end{figure*}
 
The quantum logic circuits for WAVES are shown in Fig.~\ref{figure1}b and Fig.~\ref{figure1}c.
The search (Fig.~\ref{figure1}b) proceeds by preparing trial states $|\Psi \rangle_{T}$ in the target register, according to the ansatz, 
and setting the control qubit to $|+\rangle_C$.
The combined state $|+\rangle_C \otimes |\Psi \rangle_{T}$ is then evolved through a controlled unitary ($C\hat{U}$) operation which embeds the unitary $\hat U=e^{-i \hat{H} t}$ for the evolution of $|\Psi \rangle_{T}$ according to the Hamiltonian $\hat{H}$, for a time $t$. The emerging control qubit state $\rho_C = {\rm Tr}_T(\rho)$ is then analysed by single-qubit state tomography. It is thus possible to calculate the von Neumann entropy $\mathcal{S}(\rho_T)=\mathcal{S}(\rho_C)$. 
The entropy acts as an \emph{eigenstate witness}: it is zero if the target state is an eigenstate of the Hamiltonian. In particular, for small $t$ the von Neumann entropy, along with the linear entropy, are upper and lower bounded by monotonic functions of the support of $|\Psi\rangle$ in the eigenbasis of $\hat{H}$, i.e. they are sensible measures of such support, see Supplementary Materials (SM) 1.

\begin{figure*}[ht!] \begin{center}
\makebox[\textwidth][c]{\includegraphics[width=1\textwidth]{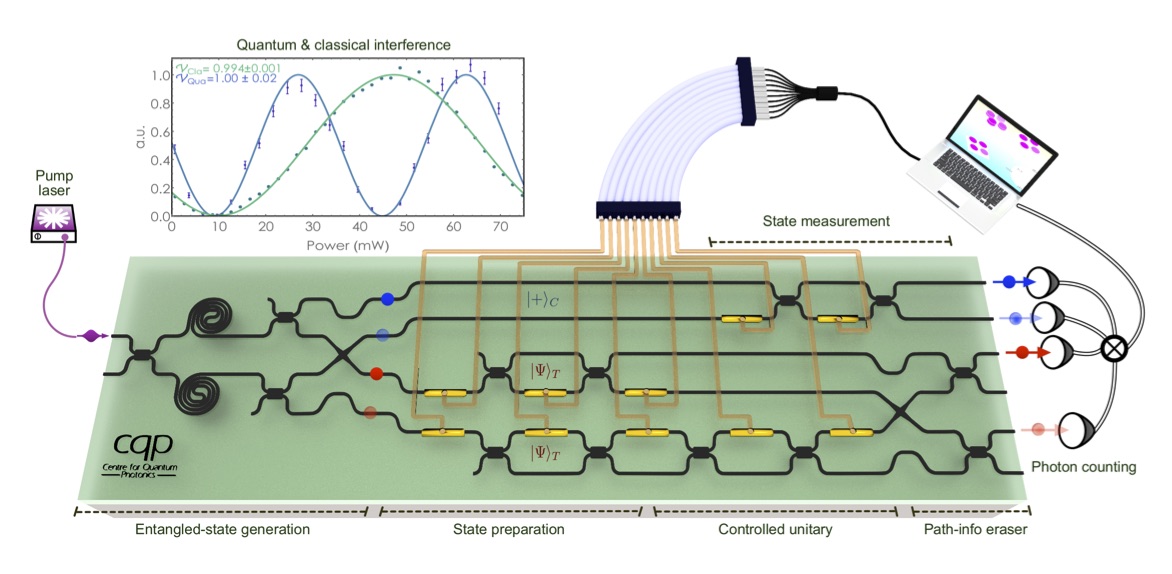}} 
 \caption{\textbf{Silicon quantum photonic processor.} 
 The quantum device enables one to produce maximally path-entangled photon states, perform arbitrary single-qubit state preparation and projective measurements, and more importantly perform any $C\hat{U}$ operation in the two-dimensional space.  
Photons are guided in the silicon waveguides and controlled by thermo-optical phase shifters.  Photon pairs are directly generated inside the silicon spiral sources through SFWM, off-chip filtered and post-selected by AWG filters (not shown), and measured by SNSPDs. The generated signal (blue) and idler (red) photons are different in wavelength, and form the control  and target qubits, respectively.  
The quantum chip is interfaced with a classical computer. Inset: high-visibility quantum (blue) and classical (green) interference fringes obtained in the device using the photon sources part and configuring the top final interferometer.  
The high visibility is essential to verify the high-performance  and correct characterisation of the device.    
} \label{figure2}
 \end{center}\end{figure*}

This measurement of the entropy enables us to variationally target excited states as well as the ground state. 
The control qubit also provides an energy estimator $\mathcal{E}=-\textit{Arg}[\langle \Psi|e^{-i \hat{H} t}|\Psi\rangle_T]/t$, evaluated using the off-diagonal elements of $\rho_C$. 
The variationally optimal ground state simultaneously minimises the entropy $\mathcal{S}(\rho_C)$ and the energy estimate $\mathcal{E}$.
Computationally, the task of finding the ground state can be interpreted as an optimisation problem using a physically motivated objective function $\mathcal{F}_{\textnormal{obj}}$ analogous to a
\textit{free energy}, $\mathcal{E}+T \mathcal{S}(\rho_C)= \mathcal{E} -T~ {\rm Tr}[\rho_C \ln(\rho_C)]$, \textcolor{black}{ where $T$ is a parameter that trades off between energy optimization and entropy optimization.}
We call $\mathcal{P}= {\rm Tr}[\rho^2_C]$ the purity of the reduced density matrix of the control qubit, which consents to approximate the von Neumann entropy with the linear entropy $1-\mathcal{P} $ in order to use the more practical objective function:  
\begin{equation}\label{eq:objfun}
\mathcal{F}_{\textnormal{obj}}(\mathcal{P},\mathcal{E})= \mathcal{E} -T \mathcal{P}= \mathcal{E} - T~{\rm Tr}[\rho^2_C] ,
\end{equation}
up to negligible constants.
The optimisation of $\mathcal{F}_{\textnormal{obj}}$ also permits one to identify excited states, because they occupy local minima in the high $T$ limit for almost all evolution times $t>0$ (SM~1). 
 Defining an initial reference state $|\Phi\rangle$ (usually obtained by mean-field approximations) and the complex vector $\vec\theta$ as the list of parameters describing the ansatz-based state preparation $\hat A(\vec\theta)$, i.e. $|\Psi\rangle_T= \hat A(\vec{\theta})|\Phi\rangle$, our algorithm proceeds as follows:
 
\begin{enumerate}
  \item Variationally search for the state parameters $\vec\theta_g$ that minimize the objective function $\mathcal{F}_{\textnormal{obj}}$, thus obtaining the unitary for the ground state $\hat A_g =  \hat A(\vec \theta_{g})$ . 
  \item Construct a unitary for an approximate $i^{\textnormal{th}}$ target excited state via $\hat E_{p_i} \hat A(\vec \theta_g )$, with $\hat E_{p_i}$ a system dependent perturbation. Variationally search for the $\vec \theta_{e_i}$ that minimises $\mathcal{F}_{\textnormal{obj}}$ in the high $T$ limit (entropy), obtaining the unitary for the target excited state $\hat A_{e_i} = \hat E_{p_i} \hat A(\vec \theta_{e_i})$.
  \item Using the $\hat A_g$ for the ground state or $\{\hat A_{e_i}\}$ for the excited ones in the state preparation, perform the IPEA which further projects each state onto the closest eigenstate \cite{Li2011} and refines the energy estimate.
\end{enumerate}
In this paper, we adopted a swarm optimisation method for the experimental variational searches, where for each iteration $\mathcal{F}_{\textnormal{obj}}$ is measured for a swarm of trial states (\textit{particles}), randomly sampled from a \textit{prior} distribution, and used to infer a \textit{posterior} with lower $\mathcal{F}_{\textnormal{obj}}$ (for more details on the optimisation method see SM~2).
The computational resources required by our variational method to learn eigenvalues of the Hamiltonian can be quantified by the number of controlled unitary operations performed in the simulation, which depend on the optimisation method used. Further breakdown in the cost estimates can be considered by decomposing the controlled unitary into fundamental gates using Trotter--Suzuki or Linear--Combination based methods \textcolor{black}{\cite{reiher2016elucidating}}, but here we ignore such issues for simplicity. For the particle swarm gradient-free optimisation, the computational cost of sampling from the eigenspectrum of $\hat{H}$,
is described by theorem~1. Whereas \textcolor{black}{ a version} for gradient based methods is reported in the Methods~1 \textcolor{black}{(both demonstrations can be found in SM~3)}.

\begin{theorem}\label{thm:1}
Let $\hat{H} \in \mathbb{C}^{2^n \times 2^n}$ be Hermitian and assume that after $k\in \{1,\ldots,N_{\rm iter}\}$ epochs the state $\ket{\psi_T(k)}= \sum_i \alpha_i(k)\ket{\lambda_i}$, where $\hat H\ket{\lambda_i}=\lambda_i \ket{\lambda_i}$ for $\lambda_i\ge 0$ and that the sequence of sets of particles $\{\Xi(k):=\{\vec{\theta}_j\}\}$
 satisfies 
 $\max_{\vec{\phi} \in \Xi(k)}\|\vec{\phi}-\mathbb{E}_{\vec{\theta} \in \Xi(k)} (\vec{\theta})\|_{\max}\le x_{\max}$ and $dim(\Xi(k))= N$ $\forall k$.  
 Then the number of applications of controlled $e^{-i\hat H t}$, for $[0,\pi/(2 \|\hat H \|))\ni t\in \Theta( \| \hat H \|^{-1})$, required for our particle swarm optimisation algorithm to learn an eigenvalue within error $\epsilon$ with probability at least $1/2$ is in
$$O\left(N_{\rm iter}N{\rm dim}{(\vec{\theta})} \left(\frac{\|\hat H \|^2}{\min_k \sum_i |\alpha_i(k)|^4}+T^2 \right)\left[\frac{\Gamma}{\delta} \right]^{2}+\frac{1}{\epsilon}\right)$$
where $\delta$ is the maximum error in the evaluation of $\mathcal{F}_{\rm obj}$ allowed and $\Gamma:= \max_k(x_{max}(k)/\epsilon_\mu(k), x_{\max}^2(k)/\epsilon_\Sigma^2(k))$, with $\{\epsilon_\mu^2(k)\}$ ($\{\epsilon_\Sigma^4(k)\}$) the corresponding tolerance in the (variance of the) trace of the covariance matrix of the sample mean. In particular, the $O (1/ \epsilon)$ term refers to the accuracy targeted by the phase estimation part of the protocol.
\end{theorem}

The above theorem implies that \textcolor{black}{the number of applications of the controlled unitary  required} only scales implicitly with the number of spin-orbitals ($n$) in the system. 
In this regime, the relevant scaling parameter for iteration cost is the dimension of the parameter space. 
\textcolor{black}{The problem of finding an appropriate ansatz is beyond the scope of this work: it is expected though to be polynomial in the number of spin-orbitals for many physically relevant systems~\cite{Romero:2017vp,Jarrod2016, Wecker2015, Georgescu:2014um, AspuruGuzik2005}.}
Similarly, the number of swarm particles required ($N$) and the number of variational steps ($N_{iter}$) depend on both the dimension of the relevant parameter space and prior knowledge about the solution.  Because the particles are moved towards the true model as the algorithm learns, $N$ is expected  to scale polynomially~\cite{beskos2014stability} for problems such as chemistry, where a good ansatz and high degree of prior knowledge is possible. 
 The number of variational parameters together with the number of swarm particles required for these specific ans{\"a}tze to achieve chemical accuracy will likely require empirical studies to be precisely estimated. 
 If Trotter methods are also taken into account for the simulation then there is a factor of roughly $n^{5.5}$ multiplied by the above costs~\cite{poulin2015trotter}.

Another fundamental point is how to choose the excitation operators used in the excited state variational search.
 Consistent choices for the system and state-specific perturbing unitaries $\hat E_{p_i}$, required to construct the excited states,
can be inferred from readily computable properties of the simulated system \cite{Jarrod2016b}. General many-body Hamiltonians for interacting particles decompose into $\hat H = \hat H_0 + \hat V$, where $\hat H_0 = \sum_i \epsilon_i \hat a_i^\dagger \hat a_i$ is a one-particle term and $\hat V$ is an interaction term. Because $\hat H_0$ dominates $\hat H$, a transition from the ground state to an excited state can be approximated by the action of a sequence of single excitation operators $\hat a_i^\dagger \hat a_j$, each with a corresponding unitary
$\hat E_p = \exp\left[\pi/2 (\hat a_i^\dagger \hat a_j - \hat a_j^\dagger \hat a_i)\right]$. 
If excitation operators, that are based on the Hartree-Fock approximation, do not provide sufficient accuracy, then alternative approximations can be used.  Advanced methods such as multi-configuration self-consistent-field (MCSCF) approximations~\cite{helgaker2014molecular} may ultimately be needed for hard instances with substantial electron correlations. In Methods~2, we  further discuss how post-Hartree-Fock methods can be used to tackle these problems through the use of natural orbitals.

\textbf{Silicon quantum photonic chip and experimental setup.}
The experimental demonstration of  WAVES  was performed on a two-qubit  
silicon quantum photonic processor schematically described in Fig.~\ref{figure2}. The device is fabricated via deep-UV lithography on a silicon-on-insulator (SOI) wafer 
with $\SI{450}{\nano \meter} \times \SI{220}{\nano \meter}$ single-mode waveguides. 
A continuous-wave (CW) pump laser in the telecom C-band with an off-chip power of approximately $\SI{10}{\milli \watt}$ is coupled into the chip using polymer spot-size converters and lensed-fibres, with a facet loss of approximately $8$~$\text{dB}$.
Pairs of single photons are generated in two $1.2$-cm-long waveguide spiral sources  through spontaneous four-wave mixing (SFWM)~\cite{Sharping2006}. Output photons are filtered using arrayed waveguide gratings (AWG) with a $\SI{0.9}{\nano \meter}$ bandwidth, spectrally selecting signal photons (blue) for the control qubit and idler ones (red) for the target. 
The photons are detected by superconducting nanowire single-photon detectors (SNSPDs) with approximately $85\%$ efficiency, obtaining a maximum photon coincidence rate of $\approx \SI{150}{\hertz}$. 
Optical interferometers consisting of thermo-optic phase shifters and multi-mode interferometer (MMI) beam-splitters are used for photonic qubit
 manipulation and analysis, driven by an electronic controller with 12-bit digital-to-analogue converters. The automation for the WAVES algorithm, including the control of quantum gates, the data collection and real-time analysis, are realised by a classical computer interfaced with the quantum photonic chip.   
More experimental details are reported in SM~5. 

Because of the low power CW pump employed in our experiment, multi-photon terms can be safely neglected. The use of the two spiral sources generates the Fock state  $(|20\rangle + |02\rangle)/\sqrt{2}$. 
High-visibility two-photon quantum interference ($\mathcal{V}_{\text{Qua}} = 1.00\pm0.02$) was observed in this experimental setup, as shown in the inset of Fig.~\ref{figure2}.  
The  photons are probabilistically split by two MMIs and then swapped by a waveguide crossing, yielding the maximally path-entangled state  $(|1010\rangle + |0101\rangle)/\sqrt{2}$ in the Fock basis~\cite{Silverstone2015, Wang2015}. 
The state of the target photon is then expanded by adding two optical spatial modes. These additional modes represent the two components of the target qubit, which is prepared in $|\Psi\rangle_T$ for both paths and undergoes an arbitrary $\hat{U}$.
That is, the operation performed on the target qubit --- either $\hat I$ or $\hat U$ --- depends on which path the photon is travelling on, indicated by $|0\rangle_P$ or $|1\rangle_P$. Path, in turn, is controlled by the state of the control qubit (the two qubits are entangled), $|0\rangle_C$ or $|1\rangle_C$, which yields a superposition of $\hat I$ and $\hat U$  gates in the circuit: 
\begin{equation}
\frac{1}{\sqrt{2}}(|0\rangle_C \otimes \hat I|\Psi\rangle_T\otimes|0\rangle_P+|1\rangle_C\otimes \hat U|\Psi\rangle_T\otimes|1\rangle_P). 
\end{equation}
By erasing the path information with the use of an additional waveguide crossing and two final MMIs and by detecting the signal photon and idler photon, we obtain the controlled unitary $C\hat{U}$ operation \cite{Zhou2011,Patel2016,Paesani2017,Wang:2017fr}
\begin{equation}
\frac{1}{\sqrt{2}}(|0\rangle_C \otimes \hat I|\Psi\rangle_T+|1\rangle_C\otimes \hat U|\Psi\rangle_T).
\end{equation}
 \begin{figure*}[ht!] \begin{center}
\makebox[\textwidth][c]{\includegraphics[width=0.75\textwidth]{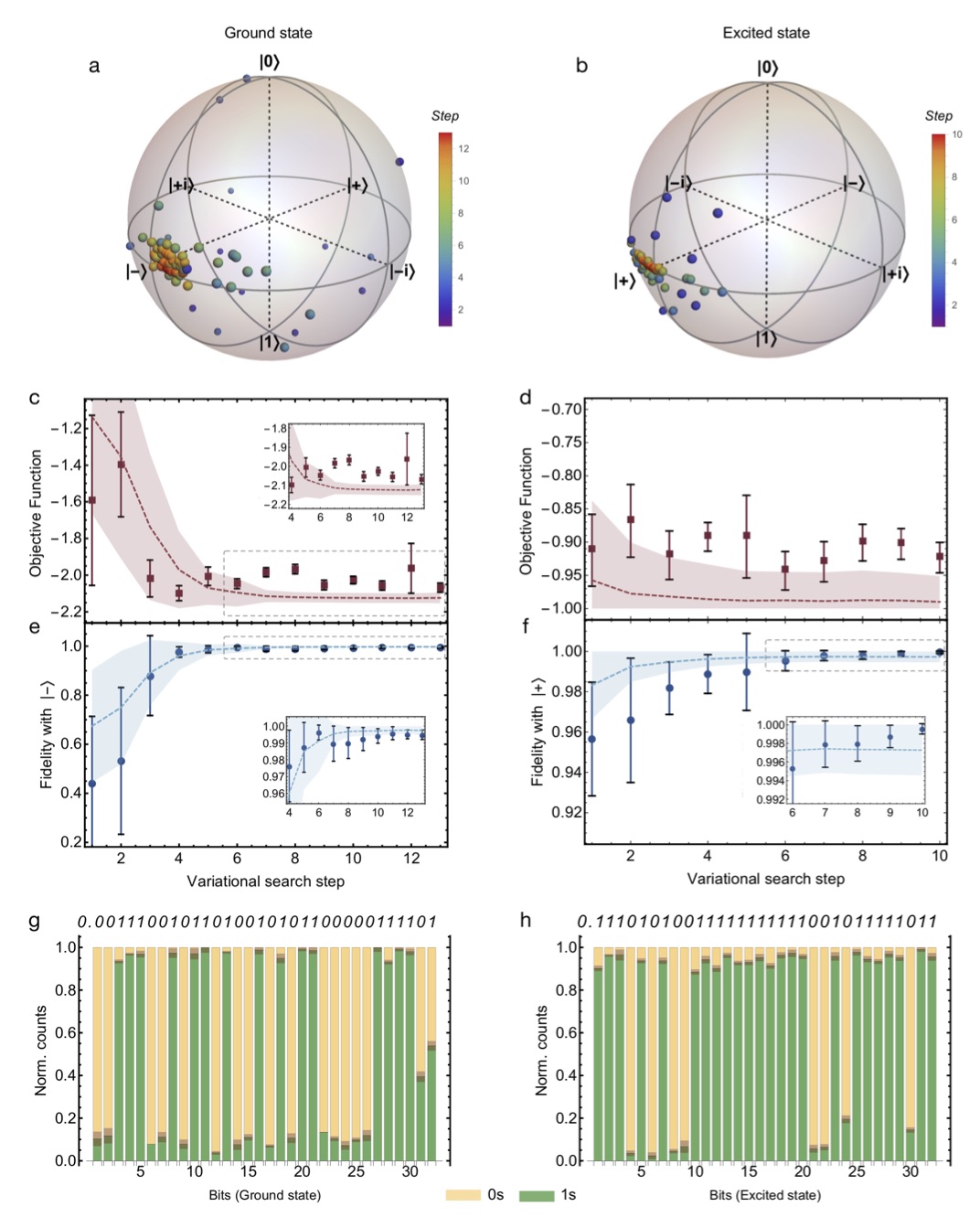}}
\caption{\textbf{Experimental results}.  
A Hamiltonian representing a single-exciton transfer  between two chlorophyll units is implemented on the silicon quantum photonic device for an experimental test of the protocol. 
(\textbf{a}) and (\textbf{b}) Colour-coded evolution of the particle swarm for the WAVES search of the ground state ($|-\rangle$) and excited state ($|+\rangle$) shown on the Bloch spheres. Different colours correspond to different steps of the search protocol. For the ground and the excited state searches we report in: (\textbf{c}) and (\textbf{d}) the evolution of $\mathcal{F}_{\textnormal{obj}}$;
(\textbf{e}) and (\textbf{f}) the fidelity ($F=|\langle\Psi|\Psi_{\text{ideal}}\rangle|^2 $) versus search steps, converging to a final value of $99.48\pm 0.28\%$ and $99.95\pm 0.05\% $ respectively. Error bars are given by the variance of the particles distribution and photon Poissonian noise.
Dashed lines are numerical simulations of the performance of the algorithm, averaged over 1000 runs, with shaded areas representing a 67.5 \% confidence interval.
Insets: behaviour close to convergence.
(\textbf{g}) and (\textbf{h}) Normalised photon coincidences used to calculate the 32 IPEA-estimated bits of the eigenphase for both eigenstates. The theoretical bit value is reported above each bar. Errors arising from  Poissonian noise are reported as shaded areas on the bars. 
}\label{figure3} 
 \end{center}\end{figure*}
 
Note that this approach implements the $C\hat U$ gate without decomposing it into multiple two-qubit gates~\cite{Lanyon2010}. 
The state preparation is realised by $\hat A=e^{i \phi_a} e^{i \phi_b \hat \sigma_z/2}e^{i \phi_c \hat \sigma_y/2}$ operations, while the $\hat U$ used to map the Hamiltonian is obtained by $\hat U=e^{i \phi_c \hat \sigma_z/2}e^{i \phi_d \hat \sigma_y/2}e^{i \phi_e \hat \sigma_z/2}$, where $\phi_i$ are  phases applied by on-chip thermal phase-shifters.
The control qubit undergoes another single-qubit operation that can be used to perform the required operations both for tomography and for the IPEA.

\begin{figure*}[ht!] \begin{center}
\makebox[1.\textwidth][c]{\includegraphics[width=1.1\textwidth]{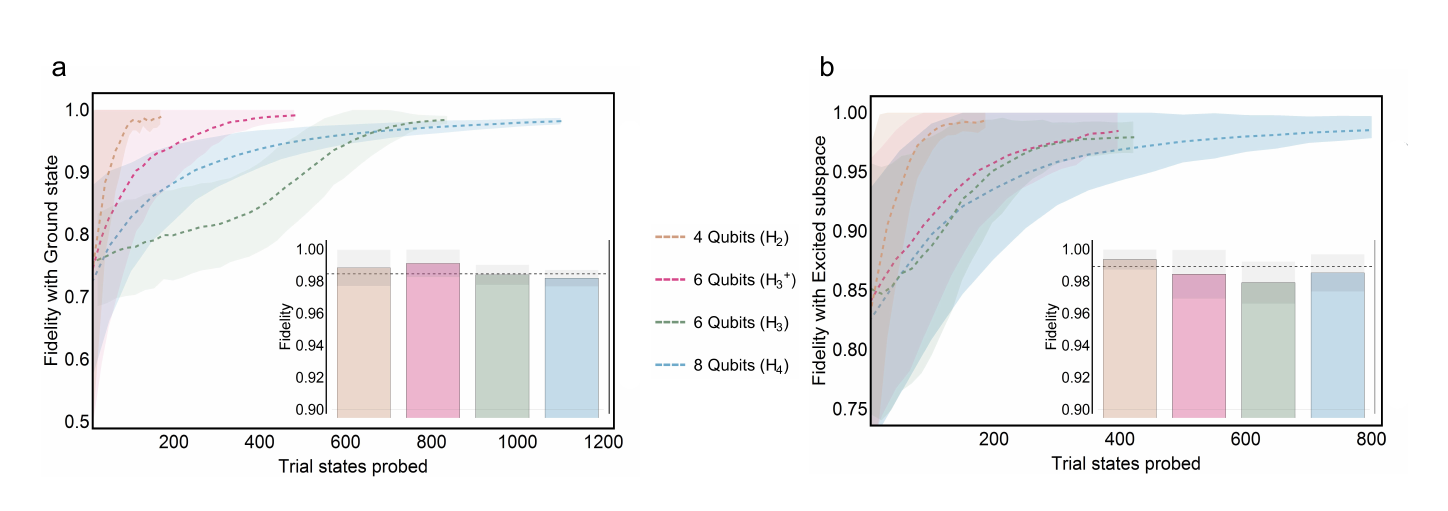}} 
\caption{\textbf{Numerical simulations for high-dimensional Hamiltonians}.  
The cases studied refer to molecular Hydrogen systems ($\text{H}_2$, $\text{H}_3^+$, $\text{H}_3$, $\text{H}_4$) with the full parametrised-Hamiltonian ansatz. 
(\textbf{a}) Variational search for the ground state of each physical system.
(\textbf{b}) Variational search for the targeted subspace of degenerate excited states with an initial excitation perturbation $\hat{E}_{p_i}$. 
On the x-axis we refer to the cumulative number of trial states probed, i.e. the number of particles in the swarm times the variational steps. For ease of comparison, the x-axis origin has been shifted in (\textbf{a}) for the various cases to have equivalent fidelity for the average initial guess. 
Dashed lines report average fidelities with the shaded areas indicating a $67.5\%$ confidence interval. 
The average fidelities achieved by the particle swarm optimisation for both ground and excited states, are calculated for 100 independent runs of WAVES. 
In all simulations a binomial noise model has been taken into account when performing projective measurements. 
Insets: bar-charts summarising final fidelities obtained by each search.
All the simulations converged to the same high fidelity within errors, as indicated by the dashed black line in the inset.}
\label{figure4}
 \end{center}\end{figure*}

\textbf{Experimental results.}
We used the quantum photonic chip to perform WAVES, calculating the eigenspectrum of a simplified exciton transfer Hamiltonian of two Chlorophyll units in the 18-mer ring of the LHII complex. 
The Hamiltonian is parametrised by $\alpha \simeq 1.46$ eV, the exciton energy of a single chlorophyll unit, and $\beta \simeq 0.037$ eV, the coupling strength between the two units~\cite{photosynthesis}. This $2\times 2$ Hamiltonian is   
written as $\hat{H} = (\alpha-\ell)\hat I + \beta \hat \sigma_x$, in the basis of Pauli operators~\cite{Jarrod2016}, where $\ell$ is a reference energy that can be chosen arbitrarily (see Methods~3). 
For this Hamiltonian the perturbing unitary for the excited state corresponds to $\hat E_p = e^{i \pi \hat \sigma_z / 2}$.

In Fig.~\ref{figure3} we show the experimental results of the WAVES approach for the ground and excited state search. The minimisation of the objective function was performed in both cases adopting the particle-swarm method outlined above.
In the experiment, the energy $\mathcal{E}$ and purity $\mathcal{P}$ used to evaluate $\mathcal{F}_{\textnormal{obj}}$ were obtained performing single-qubit tomography of the control photon.
In Figs.~\ref{figure3}a,b we report color-coded  evolution of the swarm, achieving rapid convergence of the particle distribution towards the expected eigenstates of the Hamiltonian: the ground state $|-\rangle$ 
and the first excited state $|+\rangle$. 
For the ground state search, pessimistically assuming no pre-existing knowledge of the system, the prior is initialised to span uniformly the sub-section of the Hilbert space identified by the ansatz.
For the excited state, instead, the search is initialised with the guessed state obtained by applying $\hat{E}_p$ to the ground state.

As shown in Figs.~\ref{figure3}c,d, within 10-13 search steps the $\mathcal{F}_{\textnormal{obj}}$ is converging well to its minimal value, which is corresponding to the ground state and excited state, respectively. 
 Figs.~\ref{figure3}e, f report that the mean of the particles distribution achieves a high overlap with the eigenstate targeted: fidelity of $99.48\pm 0.28\%$ with the ground state and $99.95\pm0.05\%$ with the excited one. 
All uncertainties are given by the variance of the prior distributions: 
a well motivated error bar is among the amenable features deriving from the adoption of a swarm optimisation method.

The successful convergence of $|\Psi\rangle_T$ is achieved by optimising the $\mathcal{F}_{\textnormal{obj}}$ function. 
In particular, for the ground state search we used a small value of T ($T=1.25$) in $\mathcal{F}_{\textnormal{obj}}$, while for the excited state case we used  $\mathcal{F}_{\textnormal{obj}} \equiv - \mathcal{P}$, i.e. large-$T$ limit.  
Imperfect measurements of $\mathcal{F}_{\textnormal{obj}}$, more evident in the regime close to convergence, are due mainly to experimental noise in the phase settings, given by residual thermal cross-talk among the phase-shifters. The fast convergence of the algorithm, however, indicates a good robustness of the protocol to this kind of experimental noise.

After the eigenstate search, the WAVES algorithm embeds the IPEA to improve the accuracy of eigenvalues estimate~\cite{AspuruGuzik2005, Lanyon2010}. 
In our implementation we took advantage of the circuit reconfigurability, mapping each $\hat U^{2^{k}}$ directly into the chip parameters (Fig.~\ref{figure1}c).
However, in universal quantum computers, $\hat U^{k}$ can be efficiently achieved without classical pre-compilation by cascading $k$ copies of $\hat U$~\cite{OMalley2016}. 
The IPEA estimated the binary fraction expansion of the eigenphase $\varphi (\text{mod}~2\pi)$ for both the ground and excited state energies up to $32$ bits (i.e. a precision of $2.9 \times 10^{-9}$~eV). The normalised photon counts are reported in Figs.~\ref{figure3}g,h for all the $32$ bits. Such precision is higher than what is typically achievable by spectroscopic methods.

\textbf{Numerical results for high-dimensional systems.}
We complement these proof of principle experimental results with a set of numerical simulations, providing insight into the performance of our approach  when applied to more complex Hamiltonians. 
For our numerical tests, we chose a set of molecular Hydrogen systems ($\text{H}_2$, $\text{H}^+_3$, $\text{H}_3$ and $\text{H}_4$). 
The corresponding Hamiltonians (up to 8 qubits) are represented in a STO-3G basis~\cite{Hehre:1969} in the Jordan-Wigner representation~\cite{Jarrod2016} (see Methods~4) and exhibit several degeneracies in the spectrum. For each set of degenerate excited states we will refer generically to the \textit{excited subspace} they span.

Figs.~\ref{figure4}a,b report the simulation results of the ground state search and some exemplary excited states \textit{variational} searches, respectively, 
addressing the latter ones with a set of excitation operators of the form $\hat{E}_{p_i}$. Note that this is only the first (variational) part of WAVES and that the second part (IPEA) will further project the state found into the eigenstate with a higher overlap. This feature is absent in previous VQE implementations.

For the different cases, we increased the number of particles to 8, 16, 30 and 50 for $\text{H}_2$, $\text{H}^+_3$, $\text{H}_3$ and $\text{H}_4$, respectively, which follows approximately linearly the number of parameters involved in the ``parametrized Hamiltonian'' ansatz provided in SM~7 and adopted for these simulations. 
In all the cases investigated here, WAVES is able to consistently find both the ground and excited states with high fidelities ($\simeq 99\%$ in average). The insets of Fig.~\ref{figure4} show that the final fidelities achieved by each variational search do not decrease (within errors tolerance) when increasing the size of the Hilbert space. Although these study cases do not imply scalability of the approach, they provide an encouraging result, suggesting that to keep constant the algorithm performances with the dimensionality of the problem, a sub-exponential increase in number of particles and iterations is enough, provided that a polynomial parametrisation applies.

\section*{Discussion}

\textcolor{black}{We have introduced the concept of eigenstate witness and used it to develop the witness-assisted variational eigenspectra solver (WAVES), a new quantum method for targeting both ground and excited states of a physical Hamiltonian.} 
\textcolor{black}{We showed its proof of principle implementation on a silicon quantum photonic chip for a simplified exciton transfer Hamiltonian, obtaining its eigenstates with high fidelities and estimating the eigenvalues up to spectroscopic accuracy.}
 \textcolor{black}{Additional analysis of WAVES performances is provided by numerical simulations, where the protocol yields eigenstate estimates with high fidelity for Hamiltonians of up to 8 qubits.}

All states found  using the variational search, both in the experiments and in the numerical simulations, exhibited  high fidelities with the target eigenstates. This preliminary refinement  provides IPEA with an improved approximation of the target eigenstate, leading to an exponentially higher success probability in estimating the corresponding eigenvalue and reducing the overall complexity.
Using IPEA, in addition to the variational search, allows the projection onto the eigenvectors, which is not guaranteed by the solely variational methods using a polynomial-sized ansatz. 
\textcolor{black}{As the size of the system simulated increases, the shrinking of some energy gaps may lead to eigenstates close in energy being sensed as effectively degenerate by the variational witness, within the precision achievable by the experimental platform of choice. In such cases, the  VQE-refined guess will exhibit consistent overlap with more than one eigenstate.} Nevertheless, careful modifications  to the  phase  estimation  procedure  may allow one to learn exponentially quickly either one of eigenvalues belonging to almost-degenerate eigenstates (see SM~4). 
In summary, the variational search acts as a state preparation stage for the phase estimation, while the IPEA step addresses the shortcomings present in the variational ansatz.

\textcolor{black}{From our study, for a particular ansatz (e.g. unitary coupled cluster) and using low order Trotter formulas, the time required by WAVES is expected to explicitly scale with the problem size as $O(Mn^{5.5})$, where $n$ is the number of spin-orbitals and $M$ is the number of variational parameters, in accordance with simulation results reported} in~\cite{reiher2016elucidating} for systems difficult to simulate classically such as Nitrogenase. This does not automatically imply that WAVES or any other eigenstate preparation method is efficient for any Hamiltonian, because that would imply that $\NP \subseteq \BQP$~\cite{cipra2000ising}, which is widely believed to be false. \textcolor{black}{However, an optimisation based upon an eigenstate witness allows the variational algorithm to address the problem of {efficiently} estimating an eigenvalue in the vicinity of a generic targeted state, in those cases where a polynomial-sized ansatz can be provided. This problem is in general expected to be hard on classical machines, and it is challenging to solve also using traditional variational methods.}

WAVES offers 
key improvements over previous protocols.
First, \textcolor{black}{for those instances where a good ansatz is found, it can be used to locate excited states with a quantum method in a purely variational manner, in contrast to quantum-classical linear response methods~\cite{Jarrod2016b}. Such methods avoid the need for additional non-linear optimisation, but this may limit their accuracy, and they do not yet utilize quantum phase estimation to improve the final accuracy and readout precision as in WAVES. 
Furthermore, one can speculate how the eigenstate witness provides an independent test of the protocol's success}, detecting failure cases of convergence to local optima that do not represent a single eigenstate nor excited subspace (see also SM 7).
These advantages come at the cost of controlling the evolution of the target register with an ancillary qubit, which is avoidable in previous VQE proposals. Also, in WAVES,  the ability to find specific eigenstates relies on the quality of the excitation operators. Further optimisation on the objective function, e.g. including the use of an energy penalty, can in principle overcome some of these limitations.

In terms of resource costs, the use of IPEA gives a quadratic speed-up compared to standard VQE in estimating the energy of an eigenstate within a chosen precision. Such advantages are significant given the high accuracy required in quantum chemical applications~\cite{Wecker2015}.  Moreover, WAVES does not require lengthy adiabatic preparation of targeted eigenstates~\cite{AspuruGuzik2005, Du2010, Jarrod2016}, nor an increase in the number of terms of the implemented Hamiltonian 
\textcolor{black}{and a precise knowledge of the spectral gap, unlike the folded spectrum (FS) method.} 
In particular, the FS method reduces to variational optimisation of a shifted and squared Hamiltonian $H'=(H-\lambda)^2$, thus squaring both its norm and number of terms.  The choice of the shift parameter $\lambda$ can also result in accidental degeneracies in the spectrum and dramatic closing of small spectral gaps.  In the case of a problem such as quantum chemistry, this can lead to $O(n^8)$ terms formally in the Hamiltonian, drastically increasing the cost and making it cumbersome for even small instances~\cite{Peruzzo2014,Jarrod2016}.  Moreover, the ultimate accuracy of folded spectrum methods matches the accuracy of the variational ansatz used.  However WAVES corrects essentially all of these difficulties.  It retains the original norm, spectrum, and number of terms in the Hamiltonian $O(n^4)$, does not depend upon a shift parameter, and exceeds the accuracy of the variational ansatz used through projective phase estimation.  \textcolor{black}{In SM~6 and SM~7, the interested reader can find numerical simulations comparing WAVES with previous VQE implementations, and with the folded spectrum method.}
The results show, in particular, that the FS method finds states with poor overlap with any true eigenstate in cases exploiting its weaknesses, whereas WAVES provides estimates exceeding $99 \%$ fidelity with a correct eigenstate in all cases tested. This direct comparison indicates higher reliability for WAVES, adding to the improvements in terms of resource costs.

\textcolor{black}{In conclusion, WAVES is a new approach to tackle the search for both ground and excited states of physical Hamiltonians. 
The analysis performed shows that the method is expected to be scalable, under the assumption that a good ansatz can be found. The experimental demonstration on a quantum photonic chip and numerical simulations show the method performance on small scale scenarios, indicating good noise resilience properties and better performances if compared to previous approaches.}
Our algorithm is in principle amenable to short circuit depths and leverages methods known to exhibit error robustness, thus enabling near-term experiments on non-fault tolerant machines.
\textcolor{black}{By introducing new objective functions for variational algorithms, this protocol opens the way to the investigation of new methods for computing Hamiltonian spectra, and represents a promising tool for future developments of quantum simulation on quantum computers.}

\section*{Materials and Methods}\label{Methods}

\footnotesize
\pseudosection{1. Gradient based methods computational cost of WAVES. }
In the following theorem the computational cost for the case of gradient based methods is reported.  Proof is given in the SM 3.1.

\begin{theorem}\label{thm:2}
Let $\hat H\in \mathbb{C}^{2^n \times 2^n}$ be Hermitian and assume that after $k\in \{1,\ldots,N_{\rm iter}\}$ epochs the state $\ket{\psi_T(k)}= \sum_i \alpha_i(k)\ket{\lambda_i}$ where $\hat H\ket{\lambda_i}=\lambda_i \ket{\lambda_i}$ for $\lambda_i\ge 0$.
Further, assume there exists a numerical differentiation formula that evaluates $\partial_{\theta_q} {\mathcal{F}}_{obj}$ using a constant number of function evaluations on a grid of spacing $h>0$ within error at most $\kappa (h\Lambda)^p$ for positive $\kappa,\Lambda$ for $p\in \Theta(1)$.
 Then the number of applications of controlled $e^{-i\hat H t}$, for $(0,\pi/(2 \|\hat H\|))\ni t\in \Theta( \|\hat H\|^{-1})$, required in the algorithm is in

\begin{equation*}
O\left(N_{\rm iter}\kappa^{\frac{2}{p}}\Lambda^2{\rm dim}{(\vec{\theta})} \left(\frac{\|\hat H\|^2}{\min_k \sum_i |\alpha_i(k)|^4}+ T^2 \right)\left[\frac{{\rm dim}(\vec{\theta})}{\delta} \right]^{\frac{2p+4}{(p+1)}}+\frac{1}{\epsilon}\right)
\end{equation*}
where $\delta$ is the maximum error in the two-norm of the gradient of $\mathcal{F}_{\rm obj}$ allowed and $\epsilon$ is the maximum error allowed in phase estimation of the final system with probability $1/2$. 
\end{theorem} 

 It is then clear from the analyses contained in theorem \ref{thm:1} and \ref{thm:2} that the particle swarm method has the potential to outperform the gradient based method in cases where many parameters are required to describe the ansatz state and $\Gamma$ is modest.  However, the rate at which the two learn can differ substantially, because the same number of iterations may provide more or less information than the other case. 
 In practice, gradient based methods may be more practical to find an optimal solution in the vicinity of local optima, whereas global methods like our particle swarm method may provide a better method for approaching them. Because the scaling of the Bayesian optimisation approach with the number of variational parameters is better than the bounds that we prove for gradient-based optimisation, we assume that such approaches will be better in high-dimensional problems. Moreover, being inspired by ideas from approximate Bayesian inference, the latter retains part of their noise robustness. For these reasons, we focused on the particle swarm method for experiment and simulations in this work. 

 Finally, while both methods scale quadratically with $T$, in practice the scaling will not typically be so bad.  If $\delta$ is chosen to guarantee fixed relative error for the process then the cost approaches $T^2/\delta^2$, which is constant. This means that the quadratic scaling of $T$ is not necessarily problematic in cases where the WAVES algorithm is optimising for purity.

\pseudosection{2. Excitation operators for chemical Hamiltonians.}
Our method for locating excited states variationally utilizes approximate excitation operators to enhance the rate at which excited states may be located.  Quantum chemistry has a long history of utilising the theory of linear response to external perturbations to approximate excited states of the system~\cite{watermann2014linear}.  The accuracy of this approximation relies upon the partitioning of the total Hamiltonian into $\hat H = \hat H_0 + \hat V$ where $\hat H_0$ is a non-interacting Hamiltonian of the form $\hat H_0=\sum_{ij} h^{ij} b_i^\dagger b_j$ and $\hat V$ is an interacting perturbation.  In quantum chemistry, this partitioning is often taken to be $\hat H_0 = \hat F$
where $F$ is the Fock operator that includes one-body and averaged two-body interactions, and $V$ is the remainder.  For many systems $\hat V$ is small enough such that a perturbation treatment suffices~\cite{MollerPlesset:1934}.

As a non-interacting Hamiltonian, $\hat H_0$ may be efficiently diagonalized by a unitary transformation such that $\hat H_0 = \sum_i \epsilon_i a_i^\dagger a_i$, where $\epsilon_i$ are the eigenvalues of the free-fermion Hamiltonian.  In this model, excited states may be formed through excitation operators of the form $a_i^\dagger a_j$ acting on the ground state where $j$ indexes sites currently occupied by electrons, and $i$ indexes unoccupied sites.  If $\hat V$ is comparatively small, these eigenstates will approximate eigenstates of the true Hamiltonian, and one may refine estimates within the single particle approximation space by diagonalising the Hamiltonian in the basis of vectors $\{a_i^\dagger a_j \ket{\Psi}\}$ where $\ket{\Psi}$ is some reference state.  This method is called the configuration interaction singles (CIS) method.  The connection may also be seen in the context of the first order time-dependent response to an external field.  For quantum computers, variations of these states may be prepared by the unitary operators $\hat E_{ij}(\theta) = \exp \left[ \theta (a_i^\dagger a_j - a_j^\dagger a_i) \right]$. 

In the case of weak interactions, classical methods such as coupled cluster have been successful in describing the ground state, however even low-lying excited states in these systems may exhibit correlation structures and entanglement that prevent their efficient description.  This is reflected in their difficulty of simulation by current classical methods~\cite{serrano2005quantum,Crawford_2000}, and represents a key motivation for quantum methods such as WAVES to study excited states.  Moreover, we stress that the single excitations here represent initial guesses for WAVES to search through correlated states not accessible to classical simulation, and that these single excitations may be derived from a reduced density matrix, using a procedure described below, that was not accessible classically due to quantum correlations in the ground state.  The WAVES method refines these initial guesses through optimisation, and then projects to a dominant eigenstate by utilising phase estimation.

In quantum computing, one hopes to go beyond states that are well-approximated by mean-field solutions through preparation of states with non-trivial entanglement.  In the VQE approach, these states are defined by the parametrisation of the ansatz, however, unlike in classical approaches, we may not have efficient access to full knowledge of the wave-function we are preparing.  In such cases, a possible approach to generating excitation operators is to look for the ``closest'' one-body system.  This problem defines the so-called natural orbitals in quantum chemistry~\cite{LowdinNaturalOrbitals:1956}, which are the orbitals that diagonalize the 1-electron reduced density matrix (1-RDM) of the prepared state, given by
\begin{align}
D_{ij} = \bra{\Psi}a_i^\dagger a_j \ket{\Psi}
\end{align}
that may be efficiently measured on any prepared quantum state, including those with entanglement.  As a symmetric positive semi-definite matrix, it may be diagonalized to yield a set of excitation operators $c_i^\dagger c_j$ to approximate the excited states of the interacting system.  Note that in the case of an anti-symmetric product state reference, such as that generated by Hartree-Fock, these orbitals are identical to those discussed above for $\hat H_0= \hat F$, as the canonical Hartree-Fock orbitals diagonalize the 1-RDM of a single anti-symmetric product state.

\pseudosection{3. The single-exciton Hamiltonian: Hamiltonian parameters, mapping and eigenvalues.}

Previous demonstrations of digital quantum simulation have focused almost exclusively on systems of interacting fermions such as electronic structure in molecules or the Fermi-Hubbard spin lattice model. Here we perform numerical simulations for several such cases in SM Section 7, reporting performances of WAVES in identifying correctly the eigenstates for the molecules $H_2$ to $H_4$.

However, physically interesting 
Hamiltonians are not restricted to interacting fermions and it is important to extend quantum simulation methodologies to general systems of interacting quantum particles and quasi-particles so that quantum simulation can impact on a broad range of problems relevant to physics, chemistry, biology and materials science. The spectrum of a $2 \times 2$ bosonic Hamiltonian was adopted for the experimental demonstration of WAVES in the main paper.
We therefore required a method to convert the Bosonic Hamiltonian $e^{-i\hat{H}t}$ into a sequence of unitary operations that can be implemented on a quantum computer.  This is significant because there is not a simple analogue of the Jordan Wigner transformation that maps Bosonic occupation numbers to qubits.  For example, if $\hat{H}$ had a concise Pauli--decomposition then Trotter--Suzuki formulas can be used to write $e^{-i\hat{H}t}\approx e^{-i \hat P_1 t} e^{-i \hat P_2 t} \cdots $ for Pauli operators $\hat P_1,\hat P_2,\ldots$.  
General purpose simulation methods can be employed to express $\hat{H}$ as a sum of (at most) $O(N^6)$ one--sparse matrices, provided that $\hat{H}$ does not contain interactions higher than two-body~\cite{Lanyon2010,Jarrod2016}.  
Such methods however are ill-suited for present day experiments as they require a coherent implementation of a graph colouring method, which requires additional qubits.

Notwithstanding this open challenge, we select to demonstrate our WAVES approach on the exciton transfer between two Chlorophyll units found in the light harvesting complexes of purple bacteria. In the basis of localised excitons on each Chlorophyll unit, the exciton transfer Hamiltonian is
\begin{equation}
\hat{H} = 
 \begin{pmatrix}
\alpha & \beta \\
\beta & \alpha
\end{pmatrix}\end{equation}
where $\alpha=1.46$~eV is the energy of the exciton on one of the Chlorophyll units and $\beta=0.037$~eV is the interaction between the excitons arising from the transition dipole between the two units. The qubit representation of this two-state Hamiltonian is obtained using compact mapping~\cite{Jarrod2016} and is
\begin{equation}
\hat{H}_{\textrm{qubit}} = \alpha \mathcal{\hat I} + \beta \hat \sigma^x
\end{equation}
where $\mathcal{\hat I}$ and $\hat \sigma^x$ are the usual Pauli matrices in the computational basis.

The WAVES approach performs, sequentially, a witness assisted variational search to find the eigenstates and a QPEA to obtain an accurate energy. At this point we recall a well known property of eigenfunction equations: the Hamiltonian $\hat{H}' = \hat{H} - \ell \mathcal{\hat I}$ has the same eigenstates as $\hat{H}$ and has eigenvalues $\lambda' = \lambda - \ell$, where $\lambda$ are the eigenvalues of $\hat{H}$ and $\ell$ is a constant. The parameter $\ell$ simply redefines arbitrarily the energy zero and we are free to exploit this mathematical equivalence to improve the performance of our algorithm.

In many quantum simulation applications the natural choice of energy zero results in Hamiltonians where the total energy is orders of magnitude larger than the energy differences relevant to the phenomena under investigation. This is particularly true, for example, for reaction energies in quantum chemistry and is also the case for our excitonic Hamiltonian where we are interested in the difference between the ground and excited state $2 \beta \ll \alpha$. Because QPEA requires a bitwise readout of the eigenvalue of each state of interest, any shift $\ell$ that reduces the magnitude of the corresponding energy increases the precision that can be obtained with a given length in the QPEA binary expansion.
In practice, a reasonable choice for $\ell$ may be obtained, for example, from a mean-field calculation, which can be performed efficiently on a classical computer. In other words, such an algorithm directly estimates the \emph{correlation energy} rather than the groundstate energy.  In order to mimic a realistic problem where mean-field theory provides a rather poor guess for the exact eigenvalues, we select an arbitrary value of $\ell \simeq 1.24$~eV in the experiment. 

The energy estimation in WAVES adopts the form: $\mathcal{E}=-\textit{Arg}[\langle \Psi|e^{-i \hat{H} t}|\Psi\rangle_T]/t$, and this imposes restrictions on the value of $t$, to avoid issues due to the $2 \pi$ periodicity of the $\textit{Arg}$ function. This is a limitation already known from QPEA, normally addressed by choosing $t$ small enough to prevent the algorithm from providing any eigenvalues $\mod~2\pi$ \cite{Lanyon2010}.\\
However, in the WAVES protocol additional boundaries for $t$ emerge from considerations about the $\mathcal{P}$ estimator, as described in SM 1.2. 
The span in Purity within the accessible Hilbert space also dominated the choice of the evolution time $t=26$.
It is also easy to verify that $26(\lambda_g-\lambda_e) \ne 0~\mod~2\pi$, therefore our choice satisfies all the conditions stated for $t$, concerning the value of $\mathcal{P}$ in the objective function.

\pseudosection{4. Hydrogen molecules: mapping and ansatz.}
\label{sec:hydrogensys}
In addition to the experimental verification described, we report in SM~7 numerical simulations of chemical Hamiltonians using classical computers, partially reported in Fig.~4 of the article.  In particular we have simulated ground and excited electronic states of H$_2$, H$_3^{+}$, and H$_4$ in a STO-3G basis~\cite{Hehre:1969} in the Jordan-Wigner representation~\cite{Jarrod2016} to assess the scalability of our proposal.  These represent 4-, 6-, and 8-qubit Hamiltonians respectively.  In our investigations we looked at two different ansatz.  First, we utilized a parametrized Hamiltonian (PH) ansatz, where we took
\begin{align}
\hat {U} (\vec{t}) = \exp \left[ i \left( \sum_{ij} t_{ij} (a_i^\dagger a_j) + \sum_{ijkl} t_{ijkl} (a_i^\dagger a_j^\dagger a_k a_l) \right) \right]
\label{eq:HAnsatz}
\end{align}
and allowed variation of the terms $t_{ij}$ and $t_{ijkl}$ to define the ansatz.  The $a_i^\dagger$ and $a_j$ represent creation and annihilation operators in the Hartree-Fock basis.  Variation was performed after transformation to Pauli operators via the Jordan-Wigner transformation.  In all cases, the reference state on which $U(\vec{t})$ acts is taken to be the Hartree-Fock state with the correct number of particles.  This is essentially a deformation of the original Hamiltonian, allowing one to preserve its symmetries and giving a natural connection to the original interaction structure of the problem.  We also utilized an ansatz of the form
\begin{align}
\hat{U} (\vec{t}) = \exp \left[ \sum_{ij} t_{ij} (a_i^\dagger a_j - a_j^\dagger a_i) + \sum_{ijkl} t_{ijkl} (a_i^\dagger a_j a_k^\dagger a_l - a_l^\dagger a_k a_j^\dagger a_i) \right]
\label{eq:Uansatz}
\end{align}
that has the form of an unrestricted unitary coupled cluster ansatz.  The key difference between this ansatz and the previous one, is that excitations between arbitrary orbitals are allowed, not just those found in the Hamiltonian.  The consequence of this, is that one may create or repair symmetry broken states that have been produced by some other means, allowing additional flexibility in the description of the state at the cost of more parameters.  In the following, we will refer synthetically to a parametrisation of the ansatz $\hat{A} (\vec{\theta})$, corresponding to:
\begin{align}
\label{eq:Uansatz2}
\hat {A} (\vec{\theta}) = \exp \Big[ 
\Big(\sum_i \theta_i \hat{A}_i \Big) t.
\Big]
\end{align}
Similarly, the approximate excitation operators used are defined in this basis as
\begin{align}
E_{ij} = \exp \left[ \frac{\pi}{2} (a_i^\dagger a_j - a_j^\dagger a_i) \right]
\end{align}
where we take $j$ to index the occupied orbitals of the Hartree-Fock reference and $i$ to index the occupied orbitals of the reference.

\section*{Acknowledgements}
We thank  G. Marshall, C. Sparrow, A. Montanaro, A. Laing, E. Johnston and J. Barreto for useful discussion and feedback. We thank A. Murray and M. Loutit for experimental support. 
We thank K. Ohira, N. Suzuki, H. Yoshida, N. Iizuka and M. Ezaki for the device fabrication. 
We thank S. Miki, T. Yamashita, M. Fujiwara, M. Sasaki, H. Terai, M. G. Tanner, C. M. Natarajan and R. H. Hadfield for the SPSPDs used for part of the characterisation of the device.
This work was supported by the UK Engineering and Physical Sciences Research Council (EPSRC; grant nos. K033085/1, J017175/1, and K02193/1). 
We acknowledge support from the European Research Council (grant nos. 648667, 608062, 641039, and
640079). 
J.R.M. was supported by the Luis W. Alvarez fellowship in computing sciences and by the Laboratory Directed Research and Development funding from Berkeley Laboratory
provided by the Director of the Office of Science of the U.S. Department of Energy under contract no. DE-AC02-05CH11231. 
S.M.--S. was supported by the Bristol Quantum Engineering Centre for Doctoral Training, EPSRC grant EP/L015730/1. 
X.Z. acknowledges support from the National Key Research and Development Program (grant nos. 2016YFA0301700 and 2017YFA0305200), the National Young 1000 Talents Plan, and the Natural Science Foundation of Guangdong (2016A030312012). 
J.L.O. acknowledges a Royal Society Wolfson Merit Award and a Royal Academy of Engineering Chair in Emerging Technologies. 
P.J.S. was supported by the Army Research Office grant no. W911NF-14-013. D.P.T. thanks the Royal Society for a University Research Fellowship (UF130574).

\section*{Author contributions}
R.S., J.W. and A.A.G. contributed equally to this work. 
R.S., J.W.,  A.A.G., S.P., N.W., J.R.M. and X.Z. developed the algorithm. 
R.S., J.W., D.B., X.Z. and M.G.T. designed the experiment. 
R.S., A.A.G., N.W. and J.R.M. performed simulations.
R.S., J.W., A.A.G., S.P., P.J.S. and J.W.S. performed the experiment, with theoretical support from N.W., J.R.M, S.M.S. and D.P.T.  
N.W. developed the theorems and their proofs. 
R.S., J.W., A.A.G., S.P., N.W. and J.R.M. wrote the manuscript with feedback from all authors.  J.L. O'B. and M.G.T. supervised the project.


\renewcommand{\figurename}{Supplementary Fig.}
\renewcommand{\tablename}{Supplementary Table}

\normalsize
\onecolumn
\appendix

\def\mytitle{Supplementary Materials}
\author{}
\maketitle
\vspace{-20mm}

\renewcommand{\thesection}{\arabic{section}}  

\section{Notes about the objective function}

In the paper, we introduced the objective function
\begin{equation}
\mathcal{F}_{\text{obj}}= \mathcal{E} + T S \propto b \mathcal{E}- a P,
\label{eq:fobj}
\end{equation} 
whose optimization leads the search for the eigenstates in the WAVES method. For all ground-state searches, the free parameter $T$ was fixed with an a-priori numerically optimized value (i.e. $T = 1.25$ in the $1$ qbit experiment and $T = 1$ for all the simulations with $4-8$ qubits). We preferred this option, to ease the interpretation of $\mathcal{F}_{\text{obj}}$ in terms of the energy and purity witnesses. 

 We acknowledge that accurate a-prioristic optimization of $T$ will not be possible in more complex cases. However, for such cases $T$ may be left as an additional parameter to be learned during the optimization.
This may be obtained running an adaptive version of the algorithm that was found optimal in noiseless simulations of the algorithm, whereas introducing a moderate amount of noise, adaptive runs tend to perform worse than optimal fixed choices for $T$ within the same number of steps. 

In the following paragraph, we consider and discuss in further detail the implications of the objective function introduced for WAVES as in SM Eq.\ref{eq:fobj}, showing that it has the correct properties for an eigenstate witness and justifying its adoption as an objective function. Indeed, it is not only important that $\mathcal{F}_{\text{obj}}$ achieves a local minimum at an eigenstate, but also that it behaves well in an $\epsilon$--neighbourhood about such optimal points.

\subsection{Justification of Purity in the objective function}
\label{sec:purityjust}

A key feature of our approach is to minimize an analogue of the free energy of the system in order to find parameters that make it as close as possible to an eigenstate prior to the phase estimation step.  Here we will show that the purity of the reduced density operator of the control qubit in our method can be used to bound the support of the initial target state in the eigenbasis of the Hamiltonian.

\begin{proof}
To begin, let us assume that our variational ansatz can be written as $\ket{\psi}=\sum_j \alpha_j \ket{\lambda_j}$ such that $\hat H\ket{\lambda_j}=\lambda_j\ket{\lambda_j}$ for $\hat H\in \mathbb{C}^{2^n \times 2^n}$.  Then we have that prior to the tomography step, the state can be written as
\begin{equation}
\frac{\ket{0}\ket{\psi}+\ket{1}e^{-i \hat H t}\ket{\psi}}{\sqrt{2}}.
\end{equation}
Then performing a partial trace over the second qubit register gives
\begin{equation}
\rho=
\begin{pmatrix}
\frac{1}{2} & \frac{1}{2}\left(\sum_j |\alpha_j|^2 e^{i\lambda_j t}\right) \\
 \frac{1}{2}\left(\sum_j |\alpha_j|^2 e^{-i\lambda_j t}\right) &\frac{1}{2}
\end{pmatrix}
\label{eq:TraceTarget}
\end{equation}
Matrix multiplication then gives the following expression for the purity of the control qubit

\begin{equation}
{\rm Tr}(\rho^2) = \frac{1}{2}\left(\sum_j |\alpha_j|^4 \right) + \frac{1}{2}\left(1 + \sum_{j\ne k}|\alpha_j|^2|\alpha_k|^2 \cos([\lambda_j-\lambda_k]t) \right)\ge \sum_j |\alpha_j|^4.
\label{eq:IPR}
\end{equation}

This expression is $1$ if and only if $\alpha_{j}=0$ for all $j$ such that there exists $k$ obeying $(\lambda_j-\lambda_k)t \ne 0~\mod~2\pi$.  This will generically occur for almost all $t>0$.  Thus for almost all $t$, obtaining a purity of $1$ implies that $\ket{\psi}$ is an eigenstate of the Hamiltonian.  Thus maximizing the purity is a well justified goal for our protocol.
\end{proof}

An even clearer result emerges if the evolution time is assumed to be short, i.e. $\max_{j,k} |\lambda_j-\lambda_k| t \le \pi/2$, then Jensen's inequality gives

\begin{equation}
{\rm Tr}(\rho^2)\le \frac{\left( 1 + \cos\big(\sum_{j,k} |\alpha_j|^2|\alpha_k|^2 |\lambda_j - \lambda_k|t \big) \right)}{2},
\end{equation}
the right hand side of which is a function that decreases monotonically with the average eigenvalue gap.  Similarly, Eq.~\eqref{eq:IPR} shows that the purity is lower bounded by the \textit{participation ratio}, the reciprocal of which is a frequently used estimate of the support of a state.  This shows that, under the assumptions of relatively short evolution time, the purity is upper and lower bounded by sensible measures of the support of a state, thus strengthening its adoption as a sensible measure.
One can also push beyond the limitations posed by short evolutions by replacing the maximization of the purity at a fixed time with  a maximization of the infinite time average of the purity.
\begin{equation}
\lim_{T\rightarrow \infty} \frac{1}{T}\int_0^T{\rm Tr}(\rho^2)\mathrm{d}t = \frac{1}{2}\left(1+\sum_j |\alpha_j|^4 \right).
\label{eq:integralpur}
\end{equation}
We see from this that maximizing the average of the purity over sufficiently long evolution times is equivalent to maximizing the participation ratio, the reciprocal of which gives the support of the eigenstate in the eigenbasis of the Hamiltonian.  Thus a tight connection can be made between the time averaged purity and the support of the ansatz state in the WAVES algorithm.

\subsection{Choice of appropriate \textit{t} for the experimental implementation}
\label{sec:appr_t}

We are interested in choosing an evolution time $t$, such that we obtain a lower bound far from $1$ in Supp.~Eq.~\ref{eq:IPR}, for the purity $\mathcal{P}$ spanned by all possible trial states. This is important for WAVES protocol, considering that eigenstate searches are (jointly) driven by $\mathcal{P}$. The global range of purity values will thus depend upon $\Lambda_M t= \max_{j,k} |\lambda_j-\lambda_k| t$, whereas the minimum purity observed when searching the Hilbert space between two generic eigenstates will depend upon their energy difference $\Lambda_M t = \min_{j,k} |\lambda_j-\lambda_k| t$, i.e. on the degeneracy of the eigenspectrum.\\ 
If purity were to span a small range of values in the Hilbert subspace searched, a higher accuracy in the measurements would be required in turn, in order to make eigenstate witnesses sufficiently robust against experimental noise. This could be achieved, for example, via an increase in the number of measurements to obtain each datapoint.
Therefore, $t$ shall be chosen big enough to keep the number of measurements as small as possible. At the same time, it is important to satisfy the requirement to invoke Jensen's inequality for $\mathcal{P}$ to be a sensible measure of the eigenstates' support.\\ 
Experimentally, therefore, a meaningful choice is to increase $t$ to make it as close as possible to the aforementioned upper bound $\Lambda_M t \leq \pi/2$, in order to reduce the accuracy required for $\mathcal{P}$ measurements. This can be achieved using e.g. mean field estimates of the eigenvalues $\lambda_j$ we are interested in. If for such choice $\Lambda_M t \ll \pi/2$, then either a corresponding increase in the number of measurements, or a different witness as in Supp.~Eq.~\ref{eq:integralpur} may be necessary.

\subsection{Von Neumann entropy}

After having discussed the scaling of the purity, we will now focus on showing that, for small $t$, the von Neumann entropy of the control qubit is also a well motivated measure of the support of the state.  From the fact that the linear entropy is a lower bound on the von Neumman entropy we have that if $\max_{j,k} |\lambda_j-\lambda_k| t \le \pi/2$ then

\begin{equation}
{\rm Tr}(-\rho\log(\rho)) \ge 1-{\rm Tr}(\rho^2) \ge \frac{\left( 1 - \cos\big(\sum_{j,k} |\alpha_j|^2|\alpha_k|^2 |\lambda_j - \lambda_k|t \big) \right)}{2}.\label{eq:supp7}
\end{equation}
Thus the von Neumann entropy is bounded below by a monotonically increasing function in the average eigenvalue gap under our assumptions on $t$.  
We similarly see from~\eqref{eq:supp7} and~\eqref{eq:integralpur} that the lower bound holds for almost all $t$.

We can show a similar upper bound through the use of Jensen's inequality 
\begin{equation}
{\rm Tr}(-\rho\log(\rho)) \le -\log\left(\frac{1}{2}+\frac{1}{2}\left(\sum_{j,k}|\alpha_j|^2 |\alpha_k|^2 \cos([\lambda_j-\lambda_k]t) \right) \right).
\end{equation}
If we now assume that $\max_{j,k}|\lambda_j-\lambda_k|t < 2$ then we see from the bound $\cos(x) \ge 1-x^2/2$ and the monotonicity of $\log(x)$ that
\begin{equation}
{\rm Tr}(-\rho\log(\rho)) \le -\log\left(1-\frac{1}{4}\left(\sum_{j,k}|\alpha_j|^2 |\alpha_k|^2 [\lambda_j-\lambda_k]^2t^2 \right) \right).
\end{equation}
This shows that the von Neumann entropy is also bounded above by a monotonically increasing function of the expected square of the gap, therefore it is a well motivated measure of the support of the state in the eigenbasis of $\hat H$.  \\

 \section{Swarm Optimization Algorithm}
\label{sec:swarm_algorithm}

We discuss here the optimization algorithm employed for the minimization of $\mathcal{F}_{\text{obj}}$ in our experiment, which is a gradient--free particle swarm optimization algorithm. Other optimization methods may be adopted for this subroutine of the protocol, however this swarm approach is inspired by ideas from approximate Bayesian inference, thus retaining part of their noise robustness.
The idea is to have a swarm of trial states (particles) randomly sampled from a prior distribution.
$\mathcal{F}_{\text{obj}}$ is measured for each particle, and the result is used to update the probability distribution. 
A new swarm is drawn from the updated distribution to perform further steps, until a convergence criterion is triggered. The last inferred posterior provides the estimate for the optimal state and the corresponding objective function.

We present here the details of the procedure, whose pseudo-code is in Algorithm~\ref{pseudocode}.
In the case of a ground-state search, we start sampling a swarm $\Xi$ composed of $N$ particles (i.e. a set of $N$ parameters values $\{ \vec{\theta}_i\}$) from the parameter space according to a uniform distribution (we assume no prior knowledge of the ground-state). For each of the particles, the quantum device can efficiently evaluate the two quantities $\mathcal{P}$ and $\mathcal{E}$, thus the objective function $\mathcal{F}_{\text{obj}}$ can be easily computed. In the pseudo-code, to emphasize operations performed within the quantum device, we prepend a "(Q)" to any step where such operations are invoked. Once the objective function has been obtained for all the particles, the algorithm rejects some of them ($N-S$ in our example), on the basis of the corresponding $\mathcal{F}_{\text{obj}}$, reducing the set to $\Xi'$. This subset is used to compute classically the weighted mean $\mu_{\theta'}$ and standard deviation $\sigma_{\theta'}$ of a Gaussian distribution on the parameters space.
The weights of the particles are here directly proportional to their measured $\mathcal{F}_{\text{obj}} (\vec{\theta}_i)$. 
At each step we re-sample the swarm: $N-S$ new particles are drawn from the (last updated) Gaussian distribution $\mathcal{N}(\mu_{\theta'},\sigma_{\theta'})$, and replace the ones filtered out in the new input set for the following step.   
In practice we take $S=\lceil \sqrt{N}\rceil$, but acknowledge that for more challenging examples $S\in \mathcal{O}(N)$ will likely be required.

\begin{algorithm}[ht]
\caption{Swarm optimization algorithm}
\label{pseudocode}
	\KwData{Random set of $N$ particles, characterized by the set of parameters $\Xi:=\{\vec{\theta_i}\}$, sampled uniformly from the parameters' space }
	initialization of parameters: $a, b$, \textit{threshold(s)}, $S$\;
	\While{convergence \textbf{not} achieved} {
		\ForEach{$\vec{\theta_i} \in \Xi$}{
			(Q) evaluate control qubit \textit{purity} ($\mathcal{P}$)\;
			(Q) evaluate \textit{energy estimator} ($\mathcal{E}$)\;
			calculate $\mathcal{F}_{\text{obj}} (\theta')= -a \mathcal{P} + b \mathcal{E}$\;
			}
		calculate $\mu_{Fobj}'=\langle \mathcal{F}_{\text{obj}}(\vec{\theta_i}) \rangle_{\Xi}$ \;
		\eIf{$|\mu_{Fobj}'-\mu_{Fobj}| <$ threshold for objective function}{
		convergence achieved (via $\mathcal{F}_{\text{obj}}$) \;
		}{
		$\mu_{Fobj} \gets \mu_{Fobj}'$\;
		}
		find the subset $\Xi':=\{ \vec{\theta_i'} \} \subset \{ \vec{\theta_i} \}$ of the $S$ particles with lowest $\mathcal{F}_{\text{obj}}$\;
		\If{$\mathcal{F}_{\text{obj}}$ is adaptive}{
			adjust $a \propto |\langle \mathcal{P} \rangle_{\Xi'} - \langle \mathcal{P} \rangle_{\Xi}|$ 
			and $b \propto |\langle \mathcal{E} \rangle_{\Xi'} - \langle \mathcal{E} \rangle_{\Xi}|$ \;
			normalize $\{a,b\}$ \;
		}
		calculate (weighted) mean $\mu_{\theta'}:= \langle \vec{\theta_i'} \rangle_{\Xi'}$ for the $\Theta'$ subset parameters\;
		calculate (weighted) variance $ \sigma^2_{\theta'}:= \langle (\vec{\theta_i'})^2 \rangle_{\Xi'} - \langle \vec{\theta_i'} \rangle^2_{\Xi'}$ for the $\Xi'$ subset parameters\;
		\If{$ \sigma_{\theta'} <$ threshold for particle dispersion}{
		convergence achieved (via $ \sigma_{\theta'}$) \;
		}
		generate a set $\Omega$ of $N- S$ new particles, Gaussian distributed around $\mu_{\Xi'}$ with standard deviation $\hat \sigma_{\Xi'}$\; 
		$\Xi \gets \Xi' \cup \Omega$ \;
	}
\KwResult{$\mu_{\theta'}$ is the output best estimate for the optimal parameters, and $\sigma_{\theta'}$ the error estimate. }
\end{algorithm}

In our implementation, convergence is achieved and the algorithm stopped either when the Gaussian distribution is sufficiently peaked (i.e. the spread of the distribution is lower than an accuracy determined a-priori to be satisfactory for identifying the ground state), either when fluctuations in $\mathcal{F}_{\text{obj}}$ along the last steps are smaller than a pre-determined threshold (indicating that the algorithm is drawing new particles in a plateau of the objective function, and therefore a restart is recommended). \\
The swarm optimization in the case of excited-state searches is fully equivalent, the only differences being: i) a choice of $a \gg b$ (or equivalently $T \gg 1$), ii) the initial distribution of the swarm $\Xi$ assumed as a Gaussian distribution, centred around the initial guess $\ket{\Psi_0} = \hat{E}_{p_i} \hat{A}_g \ket{\Phi}$, as provided by the appropriate excitation operator. 
Several modifications in the \textit{swarm optimization} algorithm were investigated.
One possible improvement is to adaptively choose at each algorithm step the parameters $\{ a, b \}$ concurring to the definition of $\mathcal{F}_{\text{obj}}=b \mathcal{E}- a \mathcal{P}$. This can be useful for non-convex optimizations or when poor a-priori information is available about values spanned by $\mathcal{E}$ and $\mathcal{P}$. 
A different strategy is to enhance the algorithm rapidity to converge, introducing some "greediness". This amounts to replace $\mu_{\theta'}$ with the parameters of the particle exhibiting the best $\mathcal{F}_{\text{obj}}$, $\vec{\theta'}_{best}$.

\subsection{Optimizing the number of particles for the experimental study case}
\label{sec:numeric_perfom}

Preliminary simulations were performed in order to test the algorithm behaviour for the study-case experimentally addressed in the paper. 
In particular, it was observed how an initial set of as few as $N=8,S=2$ particles was capable converging to the target state with average fidelity $F_{|-\rangle} > 0.997$, also when including a moderate level of Poissonian noise to the simulated photon counts at the output of the circuit (replacing the exact photon coincidences at the output ports of the simulated circuit, with Poissonian-distributed coincidences having a mean coincident with the numerically computed value, and a variance that assumes pessimistically a peak of $\sim 200$ experimental coincidence counts ).

$N=8$ was thus chosen as a reasonable trade-off between the final fidelity obtained and the total number of tomographies required, bearing in mind that in a quantum simulation framework, an overlap of $0.9-0.99$ with the ground state is already considered a very good starting point for a fine-grain energy estimation using IPEA~\cite{AspuruGuzik2005}.\\

\section{Complexity analysis of the variational protocol}
\label{sec:complexity}
In the article we discuss some of the main aspects concerning the computational cost of WAVES. In particular, in Theorem 1
from the Results section and Theorem 2
reported in the Methods we evaluated the computational costs neglecting the contribution due to the unitary decomposition. In the following, we report the proofs of those theorems and we further expand the discussion on the efficiency and scalability of this approach. In SM~\ref{sec:VQEscal} we also highlight some fundamental differences between WAVES and traditional VQE and we present Supp.~ Theorem~\ref{supptheor:1}, which shows how a fault tolerant quantum computer can give a quadratic improvement in such scaling.

\subsection{Complexity of WAVES for gradient-based optimization methods}
\label{sec:proofth2}
We will here delve into the proof of Theorem 2 in Methods, which states the complexity of gradient based simulation.  \\
A major caveat must be born in mind for any such proof.  
If we wish to fully bound the complexity of any variational algorithm, we need to know how many steps it will take to converge to an acceptable variational ansatz.  In general, no such bound exists.  So while the following bounds on the complexity will invariably contain terms for the number of steps involved in the variational optimization, it is very difficult a-priori to estimate such parameter, which depends upon the quality of the initial guess, as well as the ansatz used.  In practice, a small number of variational steps was sufficient for small variational quantum eigensolvers (for example fewer than $60$ are required in~\cite{Peruzzo2014}), and we found that less than $70$ steps were enough for all the simulations performed in SM~\ref{sec:numsimH} up to 8 qubits.

\begin{proofof}{Theorem 2}
There are two phases of the algorithm.  The first step involves variational optimization of the trial state $\ket{\psi_T}$ and the second involves using phase estimation to learn the energy accurately.  The cost of the second phase is trivially $O(1/\epsilon)$ with $\epsilon$ the targeted accuracy, so we only need to focus on the former cost in this proof.

Using this objective function, we have two intermediate errors that contribute to $\delta$.  Let us define $\delta_1$ to be the error in the gradient that arises from inexact gradient estimation and let $\delta_2$ be the error that arises from inexact entropy calculations.  Then using $p^{\rm th}$-order finite difference formulas on a uniform mesh of points with side width $h$, the error in a single gradient evaluation is, for positive $\kappa,\Lambda$
\begin{equation}
O([\delta_1+\delta_2]/h + \kappa(\Lambda h)^p).
\end{equation}
The error is asymptotically optimal when both terms are equal to each other which occurs, assuming $p\in \Theta(1)$, for
\begin{equation}
h\in \Theta\left(\frac{\left(\frac{[\delta_1+\delta_2]\Lambda}{\kappa} \right)^{1/p+1}}{\Lambda}\right).\label{eq:hbd}
\end{equation}
If we demand $\delta_1=\delta_2$ and that the total error in the derivative estimate is $\delta_0$ then we have from~\eqref{eq:hbd} that it suffices to choose
\begin{equation}
\delta_1\in \Theta\left(\frac{\delta_0^{1+1/p+1}}{\kappa^{1/p}\Lambda} \right).\label{eq:delta1bd}
\end{equation}

Now, to estimate $T{\rm Tr}(\rho^2)$ within error $\delta_1$, we need to perform $O((T/\delta_1)^{2})$ experiments of the form of Fig.~1.b in the main body. Supp.~Eq.~\eqref{eq:delta1bd} then implies that the total number of controlled $e^{-iHt}$ needed to perform the operation is
\begin{equation}
O\left(\frac{\kappa^{2/p} \Lambda^2T^2}{\delta_0^{2+2/p+1}} \right).\label{eq:Tbd}
\end{equation}

Consider now the number of controlled evolutions required to learn $ \mathcal{E}$ within error at most $\delta_1$.  
In order to estimate this we use that if $\|\hat H\|t < \pi/2$ then
\begin{equation}
 \mathcal{E} = -\tan^{-1}(y/z)/t,
\end{equation}
where $y= {\rm Im} (\bra{\psi_T(k)} e^{-i\hat H t} \ket{\psi_T(k)})$ and ${\rm Re} (\bra{\psi_T(k)} e^{-i\hat H t} \ket{\psi_T(k)})$.
It is then easy to see from differentiation and standard norm inequalities that if the errors in both quantities are at most $\delta_3$ then the total error in $ \mathcal{E} $ is at most
\begin{equation}
\frac{\delta_1(|z|+|y|)}{t(z^2+y^2)} \le \frac{\delta_3}{t\sqrt{z^2 +y^2}}\in O\left(\frac{\delta_3 \|\hat H\|}{\sqrt{z^2 +y^2}} \right).\label{eq:epsbd}
\end{equation}
Now since $\ket{\psi_T(k)} := \sum_j \alpha_j(k) \ket{\lambda_j}$ and since $\pi/2\ge\lambda_j\ge 0$ we have that
\begin{equation}
z^2= (\sum_j |\alpha_j(k)|^2 \cos(\lambda_j t))^2 \ge \min_k\sum_j |\alpha_j(k)|^4\cos^2(\lambda_j t) 
\end{equation}
and
\begin{equation}
y^2= (\sum_j |\alpha_j(k)|^2 \sin(\lambda_j t))^2 \ge \min_k\sum_j |\alpha_j(k)|^4\sin^2(\lambda_j t) 
\end{equation}
These results then imply that under our assumptions
\begin{equation}
z^2 +y^2 \ge \min_k \sum_j |\alpha_j(k)|^4.
\end{equation}
Substituting this into~\eqref{eq:epsbd} gives that the error in the evaluation of $ \mathcal{E}$, $\delta_1$, satisfies
\begin{equation}
\delta_1 \in O\left(\frac{\delta_3 \|\hat H\|}{\sqrt{\min_k \sum_j |\alpha_j(k)|^4}} \right).\label{eq:delta1}
\end{equation}
Therefore \eqref{eq:delta1bd} implies
\begin{equation}
\delta_3 \in \Omega \left(\frac{\sqrt{\min_k \sum_{j}|\alpha_j(k)|^4}\delta_0^{1+1/p+1}}{\|\hat H\|\kappa^{1/p}\Lambda} \right).
\end{equation}
The number of experiments needed to achieve this is $O(1/\delta_3)$, and combining this with the result in~\eqref{eq:Tbd} shows that the number of controlled evolutions needed to learn the derivative within error $\delta$ is at most
\begin{equation}
O\left(\kappa^{2/p}\Lambda^2 \left(\frac{\|\hat H\|^2}{\min_k \sum_j |\alpha_j(k)|^4}+T^2 \right)\left[\frac{1}{\delta} \right]^{2+2/(p+1)}\right).
\end{equation}

In order to compute the gradient, we need to repeat this algorithm ${\rm dim}(\vec \theta)$ times.  Furthermore, in order to ensure that the error is at most $\delta$, with respect to the $2$--norm, in the gradient of the objective function we need to compute each component of the gradient within error at most $\delta/{\rm dim}(\vec \theta)$.  Then since the gradient calculation must be repeated at most $N_{\rm iter}$ times the total cost of computing the gradients is in
\begin{equation}
O\left(N_{\rm iter}\kappa^{2/p}\Lambda^2{\rm dim}{(\vec \theta)} \left(\frac{\|\hat H\|^2}{\min_k \sum_j |\alpha_j(k)|^4}+T^2 \right)\left[\frac{{\rm dim}(\vec \theta)}{\delta} \right]^{2+2/(p+1)}\right).
\label{eq:totcost}
\end{equation}
The final cost estimate then follows by adding the $O(1/\epsilon)$ cost of phase estimation to \ref{eq:totcost}.
\end{proofof}

\subsection{Complexity of WAVES for gradient-free optimization methods}
\label{sec:proofth1}
Now that we have discussed the cost of performing the optimization using a traditional gradient descent method, we here prove the cost of our gradient-free method below, outlined as Theorem~1 in the ``Results'' section of the main manuscript. Also it should be noted that in the following discussion we require that the error is small in the trace of the sample covariance matrix and the one-norm of the vector of means.  In practice, such algorithms can typically tolerate small error in each of the variances and the max-norm of the vector of means. Nevertheless for the formal proof we adopted a more conservative metric, because if the norm is bounded with respect to the trace and the one-norm, then we can exclude pathologic behaviour in the asymptotic limit of many-parameter problems.

\begin{proofof}{Theorem 1}

The proof proceeds similarly to Theorem~2 in SM~\ref{sec:proofth2}. Again, the final phase estimation in WAVES requires $O(1/\epsilon)$ applications of the controlled unitary operation. So it suffices to focus on the variational search cost.
The optimization process requires $O(N N_{iter})$ evaluations of the objective function.  If we wish to learn the objective function within error $\delta$ then the cost of this process is from~\eqref{eq:delta1} and the discussion following~\eqref{eq:delta1bd} at most
\begin{equation}
O\left(N_{\rm iter}N\left(\frac{\|\hat H \|^2}{\min_k \sum_j |\alpha_j(k)|^4}+T^2 \right)\frac{1}{\delta^2} \right)\label{eq:Obd}
\end{equation}

The next step is to constrain $N$ such that the bounds on the variance in the sample mean of the $S$ particles accepted by the algorithm and variance of the sample variance are appropriately bounded.

Owing to the fact that the statistical errors in the sample variance are i.d., we can use the fact that variance is additive.  Thus the variance in each of the ${\rm dim}(\vec \theta)$ components of the particles is 
\begin{equation}
\frac{\sigma^2}{S}\le \frac{x_{\max}^2(k)}{S}.
\end{equation}
Thus if we are to guarantee that the trace of the covariance matrix of the sample mean at most $\epsilon_\mu^2(k)$ it suffices to take
\begin{equation}
S = \Theta\left({\rm dim}(\vec \theta)\max_k\frac{x_{\max}^2(k)}{\epsilon_\mu^2(k)} \right).
\end{equation}
Similarly the variance in the sample variance in one of the components of $\vec \theta$ is (in the asymptotic limit $S\gg 1$)
\begin{equation}
\frac{\mu_4}{S} +\frac{(S-3)\sigma^2}{S(S-1)}\in O\left(\frac{x_{\max}^4(k)}{S}\right).
\end{equation}
where $\mu_4$ is the fourth moment of one of the components.  If we then require that the variance of the trace of the covariance matrix is at most $\epsilon_\Sigma^4(k)$ then it suffices to take
\begin{equation}
S\in \Theta\left({\rm dim}(\vec \theta)\max_k\left(\frac{x_{\rm max}^4(k)}{\epsilon_\Sigma^4(k)} \right)\right).
\end{equation}
It therefore suffices to take $S\in \Theta({\rm dim}(\vec \theta)\Gamma^2)$
Then by assumption we have that $N\in O(S)$ and hence the result:
\begin{equation}
O\left(N_{\rm iter}N{\rm dim}{(\vec \theta)} \left(\frac{\|\hat H \|^2}{\min_k \sum_j |\alpha_j(k)|^4}+T^2 \right)\left[\frac{\Gamma}{\delta} \right]^{2}+\frac{1}{\epsilon}\right)
\end{equation}
for the overall method complexity follows from~\eqref{eq:Obd}.
\end{proofof}

\subsection{Efficiency}

A priori, provided that the states that emerge as part of VQE have substantial support on a small number of eigenstates, the norm of the Hamiltonian is bounded, constant precision is required and the number of variational parameters is polynomially large; WAVES is efficient given that $e^{-iHt}$ is efficiently simulatable.  First, let us focus on a particular ansatz for the state.  Consider unitary coupled cluster expansions which take
\begin{equation}
\ket{\psi_{UCC}} = e^{T-T^\dagger} \ket{\psi_{HF}},
\end{equation}
where $\ket{\psi_{HF}}$ is the Hartree--Fock state and $T=T_1+T_2+\cdots$ is the cluster operator~\cite{Peruzzo2014}.  Here the terms $T_1$ consist of single excitations away from the reference state, $T_2$ consists of double excitations etc.  Let us restrict our interest to $T\approx T_1 +T_2$.  In this case there are, for systems of $n$ spin orbitals, ${\rm dim}(\vec{\theta})\in O(n^4)$ different parameters in the ansatz.  In practice far fewer terms are required in the ansatz than this~\cite{Peruzzo2014,Wecker2015}.  Also, for constant $\|\hat{H}\|$, in practice, the norm scales not with the number of spin orbitals but with the molecular properties of the system~\cite{babbush2015chemical}.

In order to make a statement about the gate efficiency of the above methods we need to further consider a simulation algorithm for $e^{-i \hat{H} t}$.  There are many quantum simulation algorithms that can be used to perform such a simulation efficiently (for constant error tolerance).  The simplest are Trotter--decompositions, which for a system of $n$ spin orbitals require at most $n^{8+o(1)}$~\cite{wiebe2010higher} primitive exponential operations (which can each be implemented using a logarithmic number of $T$ gates in a Clifford +$T$ gate library~\cite{kliuchnikov2012fast}) for a second quantized Hamiltonian.  Empirically, low order Trotter-decompositions are found to require a number of sequential exponentials that scale roughly as $O(n^{5.5})$.  Recent work in truncated Taylor series simulations have demonstrated upper bounds on the cost that scale as $O(n^5)$~\cite{babbush2016exponentially}.  In all such cases, arbitrarily large molecules can be simulated at cost that scales polynomially with the number of spin orbitals in the problem.

From Theorem $1$, the dominant cost of the algorithm comes from the variational steps involved (assuming $\epsilon$ and $\delta$ are fixed). Since $\|\hat H\|$ depends only indirectly on the number of spin orbitals ($n$), as are $N_{\rm iter}$, $\Gamma$ and $N$, we will treat these parameters as independent since they have no explicit dependence on the dimension.  Indeed, most are simply user selected parameters. In general the dependence or not of these parameters from the size of the system needs to be verified and it is one of the main open questions concerning VQE methods. While the quantity $\min_k \sum_j |\alpha_j(k)|^4$ does depend on $n$, this dependence only comes in implicitly through the quality of the initial ansatz.  While many ansatz-based solutions, such as Hartree-Fock, CISD or UCC states, can perform well in practice~\cite{helgaker2014molecular} it is challenging to estimate a priori how it will scale for a particular molecule, so we assume that it is constant in our subsequent discussion.  In other words, we assume that the wave function has strong support over a constant number of eigenvectors.  Under these assumptions the cost of the algorithm is then simply proportional to the cost of the simulation required for the ansatz state and the phase estimation circuit.

More specifically, the cost under the assumption of constant $\epsilon$ and $\delta$ 
the cost is then given from Theorem 1 to be
\begin{equation}
O\left(
N_{\rm iter}N{\rm dim}{(\vec \theta)} \left(\frac{\|\hat H \|^2}{\min_k \sum_j |\alpha_j(k)|^4}+ T^2 \right){\Gamma}^{2}(n^{5.5} + {\rm dim}(\vec \theta)^{2.5})
+\frac{1}{\epsilon}
\right).
\end{equation}
This shows that  
WAVES can be used to provide a scalable method for drawing samples 
that are biased towards a particular eigenstate.  Furthermore, if the eigenstate is in a degenerate (or near-degenerate) manifold of eigenstates, then the algorithm will learn an eigenvalue but the variational parameters may not correspond to an eigenstate when such degeneracies are broken. In this way, our algorithm provides scalable estimates of the energy even if the gap vanishes.  
However, we emphasize that without additional assumptions we cannot prove that any specific class of chemical systems can be efficiently estimated. We discuss such caveats in greater detail below.

\subsection{Scalability of VQE}\label{sec:VQEscal}
A major criticism levied at VQE relative to IPEA based simulation is on its scalability.  To some extent neither VQE, Waves, nor any other groundstate preparation method can be scalable in general.  The reason for this is that if VQE were capable of preparing a close surrogate to an arbitrary groundstate in polynomial time then it would ensure that $\NP \subseteq \BQP$~\cite{cipra2000ising}, which is widely believed to be false under reasonable complexity theoretic assumptions.  However, the remarkable success that post Hartree--Fock methods have shown in estimating groundstate energies of molecules~\cite{helgaker2014molecular} has lead to the conjecture that physically meaningful molecules do not have groundstates that are hard to prepare~\cite{aaronson2009computational}.  Furthermore, such criticisms do not apply only to VQE.  They also apply to IPEA and classical approaches as well.

Although resolving the question of whether VQE can be used to scalably solve all electronic structure problems would constitute a substantial breakthrough in computational complexity theory,
the question of whether the algorithm itself is scalable can be addressed.  That is to say, ``what is the computational complexity of implementing the VQE algorithm without requiring rigorous guarantees of the distance between the estimated energy and the true ground state energy?'' Here we discuss this point and show that the algorithm is in principle scalable, and  show that, through the use of quantum amplitude estimation, VQE can be made much more scalable on a fault tolerant quantum computer than it would appear from prior analyses.

WAVES consists of two phases.  The first involves optimization of the training objective function. These costs are estimated in Theorem 1 and Theorem 2 which reduce to VQE when $T=0$.  Broadly speaking, provided that the state has a significant overlap with the groundstate, constant precision is required for the optimization and the number of variational parameters used scales polynomially with the number of spin-orbitals in the molecule then both methods can be applied for fixed number of iterations.  In this sense the methods are efficient and scalable, although we cannot make promises apriori about whether the resultant estimates will be accurate.

The part of VQE that remains to be estimated is the number of samples needed if iterative phase estimation is not used.  Since the above theorems assume the use of iterative phase estimation the costs cited do not reflect that of traditional VQE.  Avoiding phase estimation does not make VQE inefficient (for constant $\epsilon$) but it can under some circumstances render it impractical. Previous work  by Wecker et al~\cite{Wecker2015} highlights this, wherein they show that VQE requires a number of measurements to estimate the ground state energy that scales as
\begin{equation}
M\in O\left(\frac{\left[\sum_j |h_j|\right]^2}{\epsilon^2} \right),
\end{equation}
given $H=\sum_j h_j H_j$ for $h_j \in \mathbb{R}$ and Hermitian $H_j$ and $\epsilon$ is the desired error target.
This to no small extent justifies using phase estimation as the final step of the algorithm.

However, in principle quadratically better scalings can be achieved using a fault tolerant quantum computer.  We state this below in the following theorem.
\begin{theorem}\label{supptheor:1}
Let $H=\sum_j h_j H_j$ for $h_j\in\mathbb{R}$ and $H_j$ both unitary and Hermitian and let $\ket{\psi}=U\ket{0}$ for unitary $U$.  Then $\bra{\psi}H\ket{\psi}$ can be estimated to within error $\epsilon$ with high probability using $O(\sum_j |h_j| / \epsilon)$ applications of $U$ and $U^\dagger$.
\end{theorem}
\begin{proof}
Our algorithm proceeds as follows.  For each term in the variational ansatz we  use a form of amplitude estimation to estimate $\bra{\psi}H_j\ket{\psi}$ within error $\epsilon/|h_i|$.  We then sum the errors from these estimates to find the final answer.

These expectation values can be estimated using the Hadamard test, which yields a qubit that is $0$ with probability ${\rm Re}\bra{\psi}H_j \ket{\psi}$ or ${\rm Im}\bra{\psi}H_j \ket{\psi}$ depending on which version of the test is used.  Let $V$ be the unitary that implements the Hadamard test (which is unitary because $H_j$ is unitary), then we can define the following operators 
$I - 2V\ket{0}\ket{\psi}\bra{\psi}\bra{0}V^\dagger$ and $I - 2P_0$ where $P_0$ is the projector onto the subspace where the qubit in question is $0$.  Because $\ket{\psi}= U \ket{0}$ these operators can be implemented using multiply-controlled $Z$ gates and at most one query to $U$ and another to $U^\dagger$.  Since the Grover oracle is found by concatenating these two operators, it can be implemented using $\Theta(1)$ queries to $U$ and $U^\dagger$.

Without loss of generality, we will focus on learning the real component of the expectation value since $H_j$ is Hermitian and unitary.
The eigenvalues of the Grover oracle are then $\phi_j=\pm 2\sin^{-1}(\sqrt{[1+{\rm Re}\bra{\psi}H_j \ket{\psi}]/2})$.  We can then learn these expectation values using phase estimation, which requires with high probability $O(1/\epsilon_j)$ applications of the Grover oracle to learn the eigenvalue to within error $\epsilon_j$.

Now let us assume that we learn $\phi_j$ within error $\epsilon_j$.  It then follows that the inferred value of the expectation value satisfies
\begin{equation}
\widehat{\bra{\psi}H_j \ket{\psi}}=2\sin^2\left(\frac{\phi_j +\epsilon_j}{2} \right)-1.
\end{equation}
Since $\sin^2(\cdot)$ is a differentiable function, Taylor analysis shows that the error in $\bra{\psi}H_j \ket{\psi}$ is 
\begin{equation}
\left|\widehat{\bra{\psi}H_j \ket{\psi}}-{\bra{\psi}H_j \ket{\psi}}\right|\in O(\epsilon_j).  
\end{equation}
This shows that, up to a multiplicative constant, the error in the inferred expectation value is the error in the phase estimation step.

Now let us pick $\epsilon_j = \epsilon/|h_j|$.  It then follows, from the triangle inequality, and the fact that the errors must be multiplied by $h_j$, that the total error in the calculation of the Hamiltonian expectation value is at most $\epsilon$.  The total number of applications of $U$ and $U^\dagger$ required by the algorithm is then
\begin{equation}
O\left(\sum_j 1/\epsilon_j \right)\subseteq O\left(\sum_j |h_j|/\epsilon \right),
\end{equation}
as claimed.

\end{proof}

This scaling is significant because it shows that, in the worst case scenarios, the scaling of the algorithm is comparable (up to logarithmic factors) to IPEA being applied on top of the best known Hamiltonian simulation methods.  Indeed, it can also be preferable because it does not require a full Trotter step to be applied.  However, it does require much deeper circuits than traditional VQE, which limits its applicability in non-fault tolerant applications.

Therefore the ultimate reason for using IPEA in our algorithm is not because it provides better asymptotic scaling above VQE for estimating groundstate energies when amplitude estimation is employed.  The best reason to use IPEA over VQE is that it provides a projector onto the eigenvectors of the Hamiltonian, which is something that cannot be promised with a polynomial sized variational ansatz. This means that the final phase estimation in our algorithm should be viewed as a means of addressing shortcomings in the variational ansatz, and the VQE step should be seen as a systematic approach to addressing the state preparation problem in IPEA on a fault tolerant quantum computer.

\section{Phase Estimation Without Quantum Collapse}
A major drawback of quantum phase estimation is that, while it substantially reduces the depth of the entire inference process, it cannot be broken up into a series of short experiments. These long evolutions are needed in order to collapse a quantum state into a single eigenvector.  Here we address this issue by introducing a new approach to phase estimation that does not require the quantum state to be collapsed; instead, we collapse the state classically.

The main idea behind our approach is to use a classical learning algorithm to patch together a series of experiments performed on identically prepared quantum states to learn a single eigenvalue.  
If we have a single eigenstate fed into the IPEA circuit then the probability of measuring ``0'' is 
\begin{equation}
P(0|\lambda_j;t,\phi)=\cos^2((\lambda_j-\phi)t/2),
\end{equation}
since this is a Bernoulli experiment $P(1|\lambda_j;t,\phi) = 1-P(0|\lambda_j;t,\phi)$.  These probabilities are also referred to as the likelihood function.
Bayesian approaches to phase estimation use the likelihood function to update a probability distribution, known as a ``prior'' distribution, that describes the current state of knowledge of the eigenphase.  
Approximate methods such as the rejection filter phase estimation 
\cite{Wiebe2016} or sequential Monte-Carlo algorithms have been successfully used to solve this problem.

However, if we have an initial state of the form $\ket{\psi} = \sum_{j=1}^{2^n} \alpha_j \ket{\lambda_j}$ where $\ket{\lambda_j}$ are eigenstates of $U$ then the likelihood function takes the form
\begin{equation}
P(0|\lambda_1,\ldots,\lambda_{2^n};t,\phi)=\sum_{j=1}^{2^n}|\alpha_j|^2\cos^2((\lambda_j-\phi)t/2),
\end{equation}
If we want to apply standard approaches to learning the groundstate energy, $\lambda_1$, then we need to store all $2^n$ eigenvalues in order to perform the update properly.  This is of course intractable for large $n$ and a natural approach to combat this is to truncate the space.  

Our approach here is to truncate the space by only inferring a single eigenvalue from such experiments by marginalizing over the $\lambda_j$ that do not correspond to the ground state.  This yields a likelihood function of the form
\begin{equation}
P(0|\lambda_1;t,\phi) := 1/(2\pi)^{2^n-1}\int P(0|\lambda_1,\ldots, \lambda_{2^n}) \mathrm{d}\lambda_2 \cdots \mathrm{d}\lambda_{2^n}=|\alpha|^2 \cos^2((\lambda_1-\phi)t/2) + {(1-|\alpha|^2)}/{2}.\label{eq:likeapprox}
\end{equation}
where the parameter $\alpha$ is key for this algorithm working properly.  
Bayesian inference can be adopted to learn $\lambda_1$ using adaptively chosen experiments, such as with the heuristic proposed in \cite{Wiebe2016}. 

We then solve the inference problem by using a modified version of rejection filter phase estimation (RFPE).  The algorithm works by discretizing the prior distribution over $P$ points and then performs approximate Bayesian inference on the discrete distribution.  The posterior mean and variance are computed and then the posterior is replaced by a Gaussian distribution with the same mean and standard deviation.  The interval used is then rescaled and shifted to accommodate the new distribution.  This process is repeated until the posterior variance becomes sufficiently small.  

This process has one key advantage over exact Bayesian inference.  The ability for the algorithm to exclude hypotheses allows it to remain stable even if the support of the initial state is broad in the eigenbasis of $U$.  In order to see this, consider an experiment with two eigenvalues and assume that the prior distribution is a sharply peaked Gaussian about $\lambda_1$ with $\sigma \ll |\lambda_1 -\lambda_2|$.  If this is the case and a datum $D$ is observed that corresponds to $\lambda_2$ then the likelihood function corresponding to this hypothesis is $\cos^2((\lambda_2 - \phi)t/2)$.  Since $t=1/\sigma$, $|\lambda_2-\phi|t \gg 1$ and in general as $t$ increases will be uniformly distributed, mod $2\pi$, over $[0,2\pi)$.  This means that the shift in the posterior arisen from this hypothesis will be random and because the prior has no support over the hypothesis, such updates will not validate $\lambda_2$.  Instead these updates randomly perturb the posterior distribution for $\lambda_1$, which given $\alpha$ is chosen appropriately small will not have a net effect on the learning.  On the other hand, if the datum was taken from $\lambda_1$ then this argument no longer holds because $ |\lambda_1 - \phi|/\sigma$ will be of order unity with high probability.  Thus we expect such an inference algorithm to be capable of learning despite the presence of additional eigenvalues.

We see in Fig.~\ref{fig:2evs} that this in fact does occur.  Two eigenvalues are considered and the initial state is chosen such that there is a uniform distribution over each eigenvalue.  We choose a uniform distribution because it is maximally pessimistic.  We then repeat the inference algorithm for a fixed number of iterations and compute the minimum difference between the posterior mean and either of the eigenvalues.  We see that despite the fact that half the data comes from a different eigenvalue, the algorithm is able to rapidly learn one of the eigenvalues exponentially quickly.  The eigenvalue it outputs is random with probability 1/2 of outputting either eigenvalue.  This verifies that quantum wave function collapse is not needed to infer eigenvalues and that the data can be classically ``collapsed" by introducing biases into the approximate Bayesian inference algorithm.
 
\begin{figure}
\centering
\includegraphics[width=0.5\linewidth]{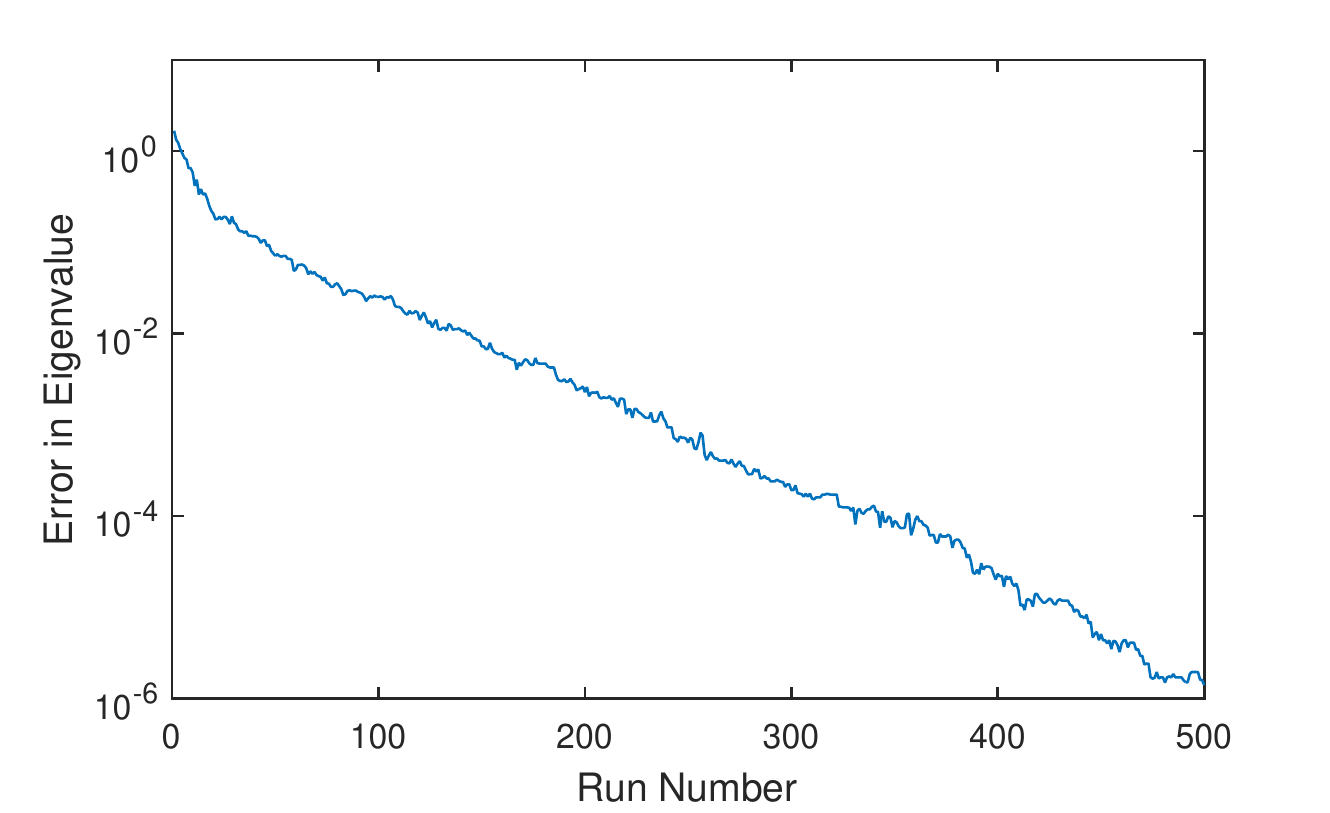}
\caption{Median inference error in eigenvalue inference with $\alpha=1/2$ for a distribution with two randomly chosen eigenvalues and likelihood function~\eqref{eq:likeapprox} is used.}\label{fig:2evs}
\end{figure}

\section{Experimental details}

\subsection{The Silicon quantum photonic device}

The silicon-on-insulator (SOI) quantum photonic chip was fabricated on a silicon-on-insulator material system using $248~nm$ deep-UV photolithography and dry etching a mass manufacturable approach. Single-mode waveguides were designed with a width of $450~nm$ and thickness of $220~nm$, and covered with a $1 \mu m$ silicon dioxide upper cladding. Spot-size converters with $300~\mu m$ long inverse taper with a $200~nm$ wide tip beneath a $4 \times 4~\mu m^2$ polyimide waveguide were used to couple photons into and out of the device. Approximately, a $8~dB$ coupling loss per facet was observed between the waveguide and single-mode lensed fiber with a $3~\mu m$ mode-diameter. 
Passive beam-splitters with near $50\%$ reflectivity were realized using multi-mode interferometers (MMIs) couplers with a footprint of $2.8~\mu m \times 28~\mu m$. The waveguide crossing in the device showed a $-40~dB$ crosstalk between two waveguides, substantially negligible for the purpose of this work. 

This device integrates key functionalities for photonic quantum information processing, such as photonic qubit quantum gates and photon-pair sources.
The two spirally-wrapped waveguide sources with a $1.2~cm$ length were designed to generate single photon-pairs by spontaneous four-wave mixing (SFWM) \cite{Sharping2006, Agrawal2006}.
The photonic qubits in this device are realised through path encoding, as shown in \cite{Shadbolt2012, Politi2008}. Such encoding allows easier reconfigurability, compared for example with polarization encoded qubits~\cite{Crespi2011}. 
The structure of the entangled source is equivalent to that one presented in \cite{Wang2015}, and is used to generate and entangled, path-encoded two qubit states, see SM~4.

The reconfigurability relies on thermo-optical phase shifters, which consists of metal resistive heaters on top of the silicon waveguides isolated by a dioxide layer in between. The heat locally dissipated in the waveguide induces the refractive index change responsible of the phase-shift. The heaters can be independently driven and controlled by home-designed multi-channel electronic. Each heater has independent connections for voltage supply but not for the ground that instead was common to all the heaters, causing electrical cross-talk between them.  
In order to drastically reduce the electrical cross-talk the phase shifters were driven in current. The heater driver board was composed of 12 bits DAC connected to amplifiers and driven through RS-232 interface. This allows to control the phase with a precision of $\simeq 0.01 ~rad$ in our experiment.
Thermal cross-talk was observed between the different degrees of freedom of the silicon quantum photonic device, the effect of which was partially compensated by performing  different calibrations for different chip configurations.

\subsection{Experimental Setup} 
\label{sec:experiment_setup}

A schematic description of the experimental set-up is reported in Suppl.~Fig.~\ref{fig:expsetup}.
The silicon quantum photonic chip was mounted and wire-bonded on an electronic PCB that was glued by thermal epoxy on a Peltier device connected to a heat-sink and controlled by a PID algorithm to stabilise the temperature of the chip. The input/output optical fibers were automatically recoupled by maximizing the coupling efficiency between waveguides and fibers.

\begin{figure*}
\begin{center}
\makebox[0.75\textwidth][c]{\includegraphics[width=0.85\textwidth]{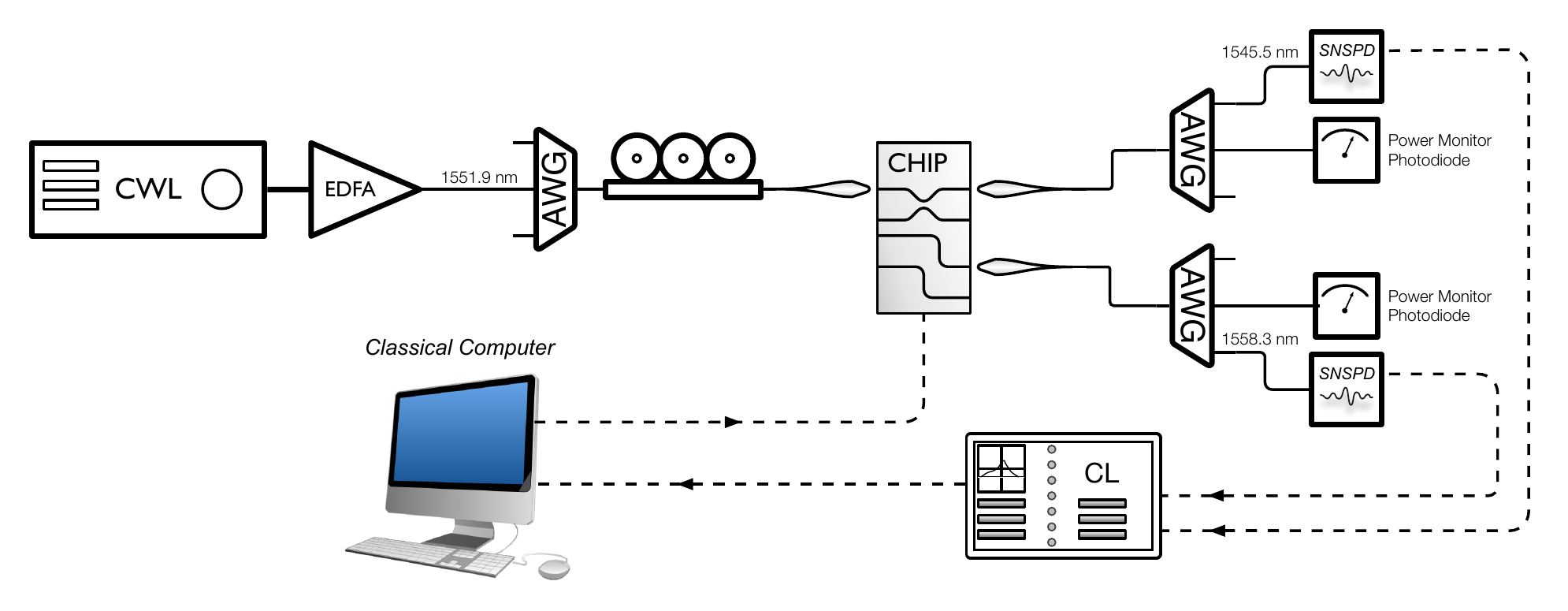}}
 \caption{ Schematic representation of the experimental set-up. The light from a CW laser light source is amplified through an EDFA. The pump is filtered through an AWG and the polarisation adjusted through a manual polarisation controller, then focused inside the chip and collected by lensed-fibres. The emerging light is then filtered through AWG, where the single photons channel and the pump are separated. While the CW light is detected by standard photo-diodes for monitoring the coupling, the single photons are detected through super-conductive nano-wire single photon detectors SNSPDs, whose pulses are analysed by a time interval analyser connected to a classical computer. The classical computer is also used to control the on-chip thermo-optic phase shifters.
 }\label{fig:expsetup}
 \end{center}\end{figure*}

A continuous-wave (CW) bright laser at $1551.9~nm$, amplified by an erbium doped fiber amplifier (EDFA), with an emerging power of  $\sim 10~mW$, was used to pump the two sources for producing photon pairs, i.e. $signal$ and $idler$ photons, by spontaneous four-wave mixing (SFWM)~\cite{Sharping2006, Agrawal2006}.
The use of a fiber-polarization controller before the silicon device ensured that the transverse-electric (TE) polarised CW light was injected into the device, after the background was removed through the use of array waveguide gratings (AWG) wavelength division multiplexers (WDM).
The output SFWM-generated photons in the quantum device were separated apart from the pump light by off-chip AWG  with a $0.9~nm$ bandwidth and $>90~dB$ extinction. We chose the signal and idler photons at the wavelength of $1545.5~nm$ and $1558.3~nm$ respectively. 
Photons were routed into single-mode fibers and finally detected using superconducting nano-wires single-photon detectors (SNSPDs) from PhotonSpot\textsuperscript{TM}, with an average efficiency of $85\%$, $100Hz$ dark counts and $50~ns$ dead-time. 
The coincidence counts were recorded by the use of a time interval analyser (Picoharp 300 by PicoQuant\textsuperscript{TM}) with an integration time of $10~s$ per data point and a coincidence window of approximately $400~ps$. 
A maximal photon-pair rate of about $150~Hz$ was observed.

The CW light and single photons were used to characterise the single-qubit gates composing the state preparation, unitary operation, and quantum state tomography. Using the CW light, we measured the mapping between the injected current and electronic power consumption, and between the electronic power and optical power, from which we linked the current and phase for each phase-shifter. A constrained least squares algorithm was used to fit the experimental data. 
The characterisation of each part of the device was performed for different configurations of the remaining parameters, in order to properly compensate for thermal cross-talk. Both $\lambda$-classical interference and $\lambda/2$-quantum interference, shown in Fig.2 inset, confirm the high quality of the device and its characterisation. Interference plots were obtained for the control qubit by setting the final MZI in the corresponding register as a Hardamard gate (i.e relative phase shift of $\pi/2$) and collecting  photon coincidences in the top two modes of the device (same procedure can be performed for the target qubit).

\subsection{Controlled unitary operations}
The implementation of WAVES requires a quantum device able to implement arbitrary controlled unitary ($C\hat{U}$) operations.
The difficulty to implement arbitrary $C\hat{U}$ gates in quantum photonics has limited first implementations of quantum algorithms to particular classes of problems with specific symmetries~\cite{Lanyon2010}.  
In this experiment we were able to implement an arbitrary $C\hat{U}$ using an entanglement-based scheme which does not require pre-decomposition~\cite{Zhou2013}. We report for the first time its adoption in integrated photonics. In this section we describe the implementation of this scheme.

Following figure 2 in the paper, the state generated by spontaneous four wave mixing (SFWM) in the two spiral sources, considering only 2 photon generation events, is given by $(|0200\rangle+|0020\rangle)/\sqrt{2}$. After the MMIs splitters  and the central crossing the full state becomes $(|2000\rangle+|0200\rangle+|0020\rangle+|0002\rangle+2 |1010\rangle+2 |0101\rangle)/\sqrt{8}$. 
Using the post-selection scheme \cite{Wang2015}, only the part corresponding to a Bell state is picked up and used for the computation.
\begin{equation}
\frac{|1010\rangle+|0101\rangle}{\sqrt{2}}
\end{equation}
The disregarded part of the state $(|2000\rangle + |0200\rangle)/\sqrt{2}$, which can be filtered out by post-selection of the output photons in the control qubit modes, was used to test the quality of the sources and circuits. By setting the tomography stage of the control qubit as the Hadamard rotation with a continue rotation around the Pauli $z$ basis $\hat R_z(\theta)$, we observed quantum interference with a visibility of $1.00\pm 0.02$ (see inset in Fig.2 of the main text).

Labelling the logic state $|0\rangle$ ($|1\rangle$) for the top (bottom) path of each photon, referring to Fig.2 of the main text, the input Fock state can be rewritten as the qubit state $(|0\rangle_C|0\rangle_P+|1\rangle_C|1\rangle_P)/\sqrt{2}$. Two additional modes are added in the bottom paths, which can be represented as the addition of another qubit, which is our target qubit, giving the state $(|0\rangle_C|0\rangle_T|0\rangle_P+|1\rangle_C|1\rangle_T|1\rangle_P)/\sqrt{2}$. For both bottom paths the two modes are prepared in the same input state $|\psi\rangle_T$, using the operations reported in the main text. The state is later evolved according to $\hat{I}$ in the upper path and $\hat{U}$ in the lower one, giving the total state
\begin{equation}
\frac{ |0\rangle_C|\psi\rangle_T|0\rangle_P+|1\rangle_C(\hat{U}|\psi\rangle_T)|1\rangle_P}{\sqrt{2}}
\end{equation}
At this point, using a waveguide crossing and two integrated beam-splitters we erase the path information between the upper and the lower path, which is equivalent to apply an Hadamard gate to the third qubit, obtaining
\begin{equation}
\frac{ (|0\rangle_C|\psi\rangle_T+|1\rangle_C\hat{U}|\psi\rangle_T)|0\rangle_P+(|0\rangle_C|\psi\rangle_T-|1\rangle_C\hat{U}|\psi\rangle_T)|1\rangle_P} {2} 
\end{equation}
Post-selecting the cases where the target photon emerges from one of the top modes, i.e. projecting in $|0\rangle_P$, we finally get the output state
\begin{equation}
\frac{ |0\rangle_C|\psi\rangle_T+|1\rangle_C\hat{U}|\psi\rangle_T } {\sqrt{2}}
\end{equation}
which gives the desired arbitrary control-unitary operation.

We remark that applying other operations after the $\hat{U}$ this scheme also allows to cascade different control-unitary operations in a non-compiled way. This can be used, for example, to perform Trotterization for the time evolution.
However the probabilistic nature of this scheme (due to the post-selection) makes it not scalable.

\begin{figure*}[h!] 
\begin{center}
\makebox[\textwidth][c]{\includegraphics[width=0.5\textwidth]{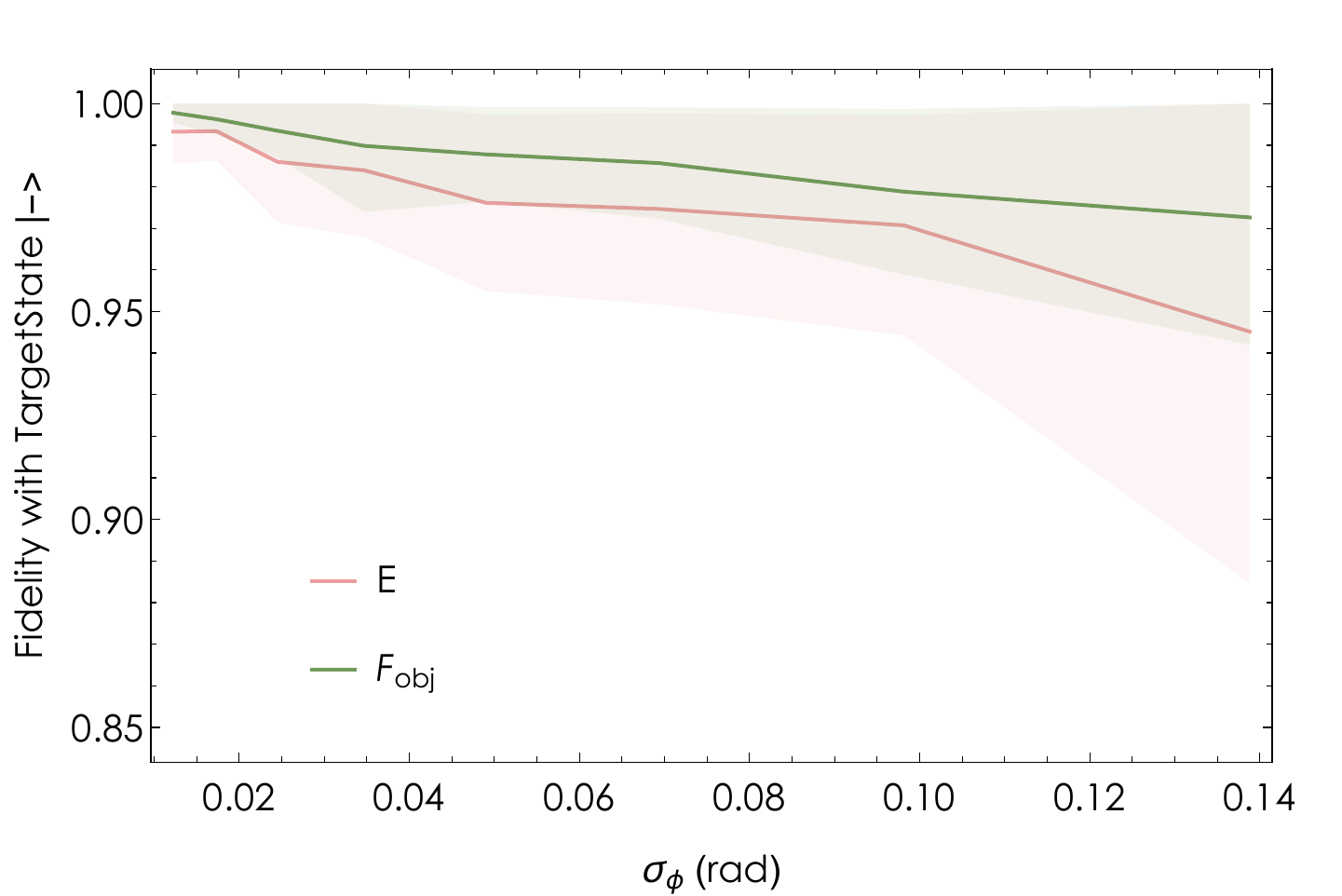}}
 \caption{\textbf{\normalsize} 
 \textsf{Numerical simulations of the variational ground-state search robustness for the bosonic 1-qubit $\mathcal{\hat{H}}$ against gate infidelities, using  $\mathcal{F}_{\text{obj}}$ or $\mathcal{E}$ alone. Infidelities in the gate implementation is modelled as Gaussian noise in the phase shifting components of the photonic circuit. Average fidelities with the true eigenstate $F_{|-\rangle}$ (solid lines) along with 67.5 \% confidence intervals (shaded areas) are reported after algorithm's convergence. }}
 \label{fig:heatnoiseplot}
 \end{center}\end{figure*}

\section{Robustness against experimental noise of the variational search}\label{sec:noise_robust}
 
In this section we investigate numerically the effect of noise in the WAVES variational search. This paragraph focuses upon the ground-state search case experimentally performed in this paper, where noise affects both estimators $\mathcal{E}$ and $\mathcal{P}$, guiding the variational search of the Hilbert space. We introduce two different sources of noise, expected to play a major role in our setup: 
\begin{itemize}
\item Poissonian noise arising from coincidence events counts (see  also SM~\ref{sec:numeric_perfom}).
\item Infidelity in the gates implementation (i.e. uncertainty in the phases implemented by the thermo-optical phase shifters on chip), affecting both the state preparation and the C${\hat{U}}$ implementation.
\end{itemize}  
Studying a noise model for this implementation of WAVES also allows to address whether the proposed objective function $\mathcal{F}_{\text{obj}}$ is capable of providing any advantage (or disadvantage) in the robustness of the search, compared to the adoption of a variational search based upon the energy estimator alone (therefore, equivalent to previous VQE ground-state search  demonstrations~\cite{Lanyon2010,Peruzzo2014,OMalley2016}). 
Data from 100 averaged numerical simulations, running the variational search with either $\mathcal{F}_{\text{obj}}$ or a simplified objective function $\mathcal{F}'_{\text{obj}}=\mathcal{E}$, are reported in Suppl.~Fig.~\ref{fig:heatnoiseplot}.\\
When Poissonian noise alone is considered discrepancies observed between the performances of $\mathcal{F}_{\text{obj}}$ and $\mathcal{F}'_{\text{obj}}$ are small, and any difference among the two approaches falls within the (statistical) error bars. For this kind of noise, therefore, nothing can be concluded in favour or against the two approaches.

Infidelity in the implementation of gates, for a digital quantum simulator without error correction, is a particularly critical feature~\cite{LasHeras2016}. In order to simulate noise affecting the qubit operations, for each thermo-optical phase shifter manipulating the target qubit we replaced the correct phase $\bar{\varphi_i}$ -- required to implement the transformation -- with a synthetic value $\varphi$, sampled from a Gaussian distribution $\varphi_i \sim \mathcal{N}(\bar{\varphi}_i,\sigma_{\varphi})$. As $\sigma_{\varphi}$ increases, the behaviour of the simulated physical circuit deviates from the ideal case. This is a realistic way to simulate such noise in our photonic setup. Each phase shift implemented is proportional to the power applied to heating element, in turn affected by electrical and thermal fluctuations. A thorough characterization is instead expected to avoid systematic biases in some of the phases $\{ \bar{\varphi} \}$. Phases involved in the control register have been excluded from this simulation. Indeed, the scope was to emphasize the role of imperfect qubit operations in the target manipulation, while neglecting contributions from errors in the tomography of the control qubit.\\ 
This noise model applied to experimental conditions led to the results reported in Fig.~3 C-F of the main text, reproducing quite accurately the experimental findings. In that case, according to device characterization data, it is assumed $\sigma_{\varphi} \simeq 0.012$ \textit{rad} in the simulations. Slight differences between the predictions and the experimental data can be likely explained via uncharacterised residual thermal cross-talk between the phase shifters (not taken into account in the model), that contributes to systematic errors in the implemented phases.
A further investigation of how disruptive such gate infidelities may be in noisier implementations is instead provided in  Suppl.~Fig.~\ref{fig:heatnoiseplot} for the bosonic Hamiltonian experimentally investigated in the main text. For this kind of noise, adopting the objective function $\mathcal{F}_{\text{obj}}$ appears a consistently more robust strategy than using $\mathcal{E}$ alone as an estimator: a difference that increases with the noise level. In particular, already for $\sigma_{\varphi} \simeq 0.14$ the final fidelity achieved in the search with $\mathcal{E}$ alone falls below $F_{\ket{-}} = 95\%$, whereas $\mathcal{F}_{\text{obj}}$ keeps providing accurate ground-state estimates.

\begin{figure}
\centering
\includegraphics[width=0.9\linewidth]{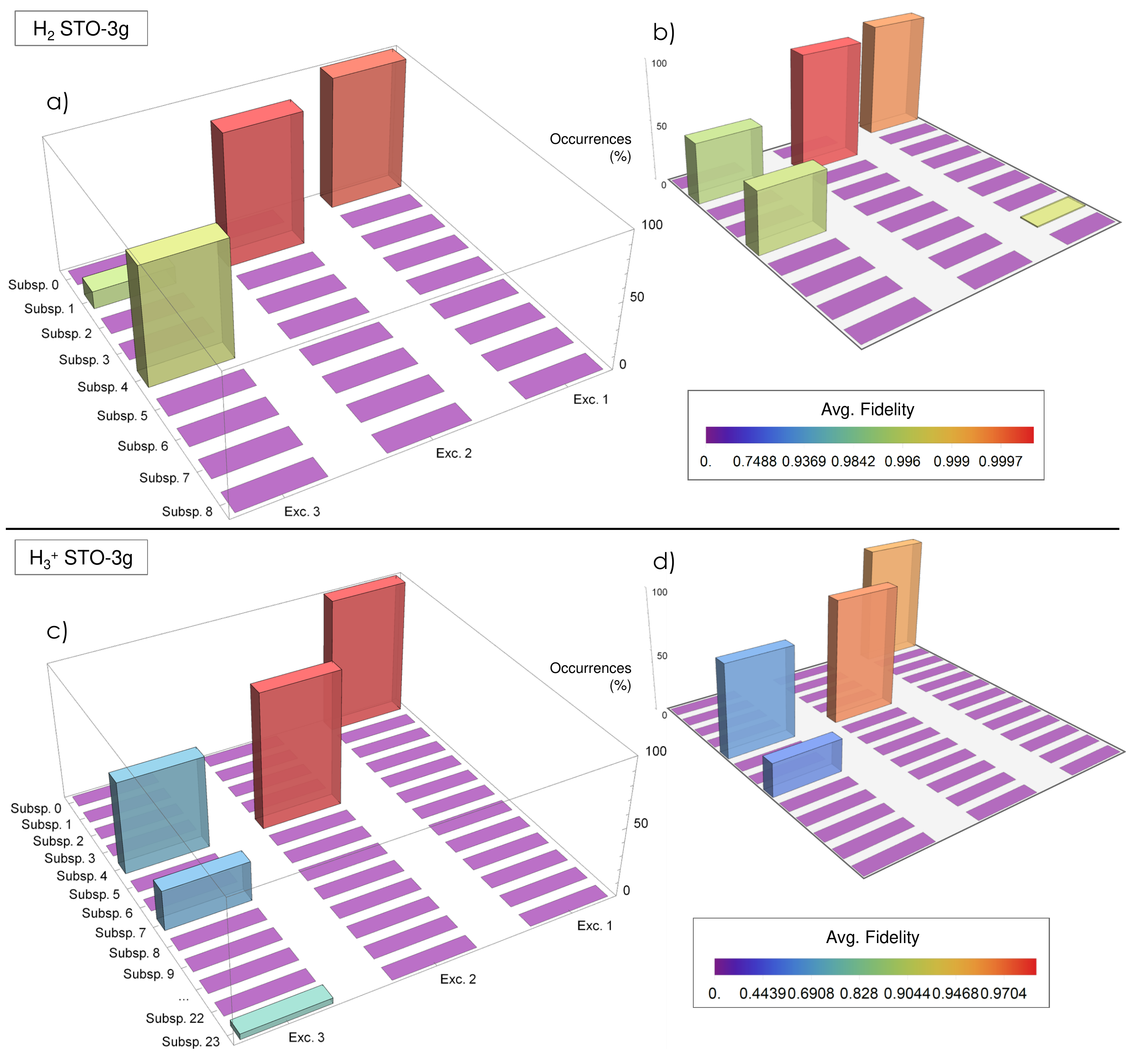}
\caption{
Synopsis of numerical simulations of excited states searches for the molecules H$_2$ and H$_3^+$, adopting the full PH ansatz (\textbf{a\&c}). Results for the synthetically reduced ansatz are in the inset outlined in light grey (\textbf{b\&d}). On one horizontal axis are reported the non-degenerate subspaces spanned by the eigenstates of the corresponding Hamiltonian ('\textit{Subsp.}~$i$'), which are $9$ for H$_2$ and $23$ for H$_3^+$. On the other axis are the 3 different excitation operators of the form $E_p$ ('\textit{Exc.} $j$') tested.
Height of the histogram with coordinates $\{i,j\}$ corresponds to the relative frequency with which the swarm of particles collapses onto the $i$-th excited subspace, when $\hat{E}_{p_j}$ is adopted in the second phase of the search. 
We accept the swarm has collapsed onto the $\mathbb{E}_i$, whenever the corresponding projector $ \prod_i $ exhibits the highest overlap with the estimate $\ket{\Psi}$ at the end of the search (see Supp.~Eq.~\ref{eq:degproj}). Fidelities reported for a certain subspace $\mathbb{E}_i$ and $\hat{E}_{p_j}$ are calculated averaging the final fidelity achieved by all and only those algorithm runs, where the swarm had collapsed onto $\mathbb{E}_i$. The bottom-right bar legend explicates the colour coding.}\label{fig:summaryplot}
\end{figure}

\section{Numerical simulations of Hydrogen molecules}
\label{sec:numsimH}

\begin{table}[!ht]	
	\begin{center} \centering{
	\begin{tabular}{c|ccccc}
		\hline
		&
		\begin{tabular}{c}
			truncated PH \\ $\dim(\vec \theta)$
		\end{tabular} 
		& 
		\begin{tabular}{c}
			standard PH \\ $\dim(\vec \theta)$
		\end{tabular} 
		& 
		\begin{tabular}{c}
			unrestricted UCCSD \\ $\dim(\vec \theta)$
		\end{tabular} 
		& 	$\dim(\{\hat{E}_{p_i}\})$ 
        & 	$\ket{\Phi}$ 
        \\ \hline
		\rule{0pt}{3ex}    H$_2$ (4 qubit)		& 13 & 15 & 110 & 3 & \ket{1100}\\
		\rule{0pt}{2.5ex}  H$_3^+$ (6 qubit)	& 60 & 66 & 462 & 15 & \ket{111000} \\
		\rule{0pt}{2.5ex}  H$_3$ (6 qubit)		& 112 & 114 & N.A. & 15 & \ket{110000}\\
		\rule{0pt}{2.5ex}  H$_4$ (8 qubit)		& 181 & 185 & N.A. & 22 & \ket{11110000}
	\end{tabular} }\end{center}
    \caption{\textbf{\normalsize} \textsf{Summary of ans{\"a}tze used for simulations in the main paper and in the SM for the various systems investigated, along with the cardinality of their parametrization - $\dim(\vec \theta)$. PH refers to the parametrized Hamiltonian ansatz in Methods Eq.~(7). We report in this same table also the number of excitation operators - $\dim(\{\hat{E}_{p_i}\})$ - that have been tested in numerical simulations and the reference state for the system $\ket{\Phi}$.}}
    \label{tab:ansatze} 
\end{table}
We expand in this paragraph the numerical simulations beyond the case experimentally tested, in order to investigate the capability of WAVES to detect different eigenstates for systems of dimensionality as high as 8 qubits. Simulations of the WAVES algorithm are here performed at a logical circuit level, unlike the simulations in SM \ref{sec:numeric_perfom} and \ref{sec:noise_robust}, where a model for the photonic circuit behaviour had been adopted for the simulations. However, all the results presented in this section include a binomial noise model for the control qubit tomography (conceptually equivalent to the Poissonian noise model for a photonic setup). The latter one was introduced in SM~\ref{sec:noise_robust}. The level of noise chosen in the simulations is compatible with what obtained in a standard experimental run for the 1 qubit case (i.e. 
$1500$ measurements for each control qubit tomography).\\
In all cases we adopt the protocol as described in Fig.~1 of the main text. That is, we first variationally optimize the parametrization $\vec \theta$ provided by one of the ans{\"a}tze in Supp.~Table~\ref{tab:ansatze} learning the optimal $\vec{\theta}_g$ to prepare the ground-state. For this search, we make no use of a-priori knowledge about the ground-state: a reference state for the system is the initial guess for the ground state, i.e. $\hat{A} ({\vec \theta_0}) \ket{\Phi} \equiv \ket{\Phi}$ in Eq.~(9) of Methods. 
We draw the initial swarm of particles from a Gaussian distribution centred around $0$ for each parameter ${\theta_i}$, with a standard deviation chosen arbitrarily to be much bigger than $\max \{ |\theta_i| \}$.
For the different molecular systems analysed, this is equivalent to having very different average initial overlaps with the ideal ground state $\ket{\Psi_g}$, as calculated diagonalizing the corresponding Hamiltonian. Precisely, $|\langle \Psi_0 |  \Psi_g \rangle|^2$ is in the range $0.7 - 0.8$ for H$_2$ and H$_3^+$, but as low as $0.2$ for both H$_3$ and H$_4$. Even when starting from such poor initial guess states, WAVES was capable of achieving fidelities higher than 95\% with the ideal target state . In those cases where a better initial overlap bootstrapped the variational search, final fidelities observed where in excess of 99\%. \\
This discrepancy in the fidelities shall not be misinterpreted as a fundamental limit deriving from the quality of the initial guess provided to WAVES. Indeed, it is worth mentioning that in order to simulate realistic experimental conditions for near-term WAVES implementations, not only we included a noise model in the control-qubit tomography, but we restricted the number of particles in the swarm as equivalently done in SM~\ref{sec:numeric_perfom}. This kept execution times for the protocol within reasonable limits for state-of-art photonic implementations.
E.g. the longest 8 qubit simulations converged after performing only about 500 single-qubit tomographies when adopting 20 particles in the swarm optimization performed over more than 400 parameters. In comparison with $\sim$100 tomographies required by our experimental proof of concept (where 8 particles in the swarm optimized over 2 parameters). 

\begin{table}[!ht]
\begin{center}
\centering{
  \begin{tabular}{c|cc}
   \hline
   \rule{0pt}{2ex}
   & Ansatz allows $\ket{\Psi_0} \rightarrow \mathbb{E}_{i (\tau)}$
   & Ansatz prevents $\ket{\Psi_0 }\rightarrow \mathbb{E}_{i (\tau)}$ \\ \hline 
   \rule{0pt}{4ex}
   \begin{tabular}{c}
      $\ket{\Psi_0}$ provides accurate guess for $\mathbb{E}_{\tau}$\\ 
      $\bra{\Psi_0} \prod_i \ket{\Psi_0} \ll \bra{\Psi_0} \prod_{\tau} \ket{\Psi_0} $, $\forall i \neq \tau$  
   \end{tabular} 
  	&  \begin{tabular}{c}
      \checkmark \\ 
      (convergence to $\mathbb{E}_{\tau}$)
     \end{tabular} 
  	& N.A. \\
   	\rule{0pt}{4ex}
   	\begin{tabular}{c}
      $\ket{\Psi_0}$ provides no accurate guess $\mathbb{E}_{\tau}$\\
      $\exists \; i \neq \tau$ : $\bra{\Psi_0} \prod_i \ket{\Psi_0} \simeq \bra{\Psi_0} \prod_{\tau} \ket{\Psi_0} $
   \end{tabular} 
  	& \begin{tabular}{c}
      \checkmark \\ 
      (convergence to either $\mathbb{E}_{\tau}$, $\mathbb{E}_i$)
  	\end{tabular} 
  	& 
  	\begin{tabular}{c}
      \ding{55} \\ 
  	(no convergence to any $\mathbb{E}_i$)
  	\end{tabular}
\end{tabular} 
}\end{center}
\caption{\textbf{\normalsize} \textsf{Summary of possible situations occurring in numerical simulations, when WAVES is performed with different ans{\"a}tze and excitation operators, targeting a certain excited subspace $\mathbb{E}_{\tau}$ from an initial guess $\ket{\Psi_0}$.}
\label{tab:performsummary}
}
\end{table}

A wide variety of behaviours can be observed when performing WAVES for excited state searches. We remind that in our protocol this is preferentially achieved by applying a set of different excitation operators of the form $\hat E_p$ to the ground state parametrization $\hat{A} (\vec{\theta}_g) \ket{\Phi}$ obtained in the first step. An additional variational search is then performed starting from this initial guess $\ket{\Psi}$, adopting $T \gg 1$ in $\mathcal{F}_{\text{obj}}$ (see Methods ``Excitation operators for Chemical Hamiltonians''). Specifically for this set of simulations where a swarm optimization is performed, the initial set of particles was initially drawn from a Gaussian distribution centred around $\theta_{i_g}$ for each parameter $\theta_{i}$, with an initial standard deviation approximately equal to $\max \{ |\theta_i| \}$. 
When targeting a subspace of degenerate excited states, we replace the fidelity $| \langle \Psi_0 | \Psi_{\tau} \rangle |^2$ for the generic trial state $\ket{\Psi}$ with a generalized expression, given by the projection onto $\mathbb{E}_{\tau}$:
\begin{align}
\label{eq:degproj}
F =
\Big\langle \Psi \Big\lvert \prod_{\tau} \Big\rvert \Psi \Big\rangle = 
\Big\langle \Psi \Big\lvert \sum_{k} (\ket{\Psi_{\tau_k}} \bra{\Psi_{\tau_k}}) \Big\rvert \Psi \Big\rangle
\end{align}
where $\{ \ket{\Psi_{\tau_k}} \}$ is a basis of $\mathbb{E}_{\tau}$. We also assume for the scope of this work that IPEA up to 32 bits of accuracy will refine the eigenvalue estimate provided by the first step of the WAVES protocol. Consequently, targeting spectroscopical accuracy, energy differences smaller than $10^{-9}$ \textit{Ha} will be ignored and corresponding eigenstates will be considered degenerate.  
Finally, in order to simplify the interpretation of the data for each variational step, we also simplistically adopt $\mathcal{F}_{\text{obj}} \equiv \mathcal{-P} \equiv \mathcal{S}-1$.
\begin{figure*}[!ht] \begin{center}

\makebox[\textwidth][c]{\includegraphics[width=0.96\textwidth]{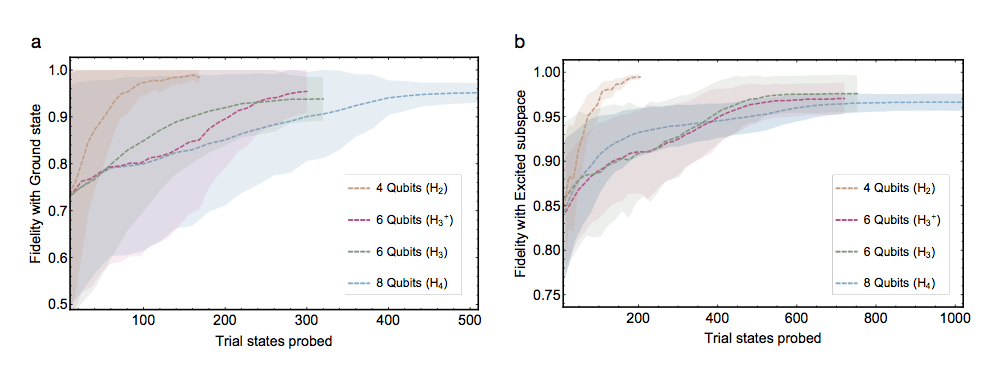}} 
 \caption{
 Numerical simulations of the WAVES variational search for the synthetically truncated PH ansatz is studied in the molecular Hydrogen systems ($\text{H}_2$, $\text{H}_3^+$, $\text{H}_3$, $\text{H}_4$).
 (\textbf{a}) Variational search for the ground-state. 
 (\textbf{b})
Variational search for the targeted subspace of degenerate excited-states with initial excitation perturbation $\hat{E}_{p_i}$. 
The average fidelities achieved by the particle swarm optimisation 
are calculated for 100 independent runs of  WAVES.  
 Dashed lines are average fidelities. Shaded areas indicate $67.5\%$ confidence intervals. 
 On the x-axis we refer to the cumulative number of trial states probed, i.e. the number of particles in the swarm times the variational steps. 
 For ease of comparison,  the x-axis origin  has been shifted in (\textbf{a}) for the various cases in order to have equivalent fidelity for the average initial guess. 
In all simulations a binomial noise model has been taken into account  when performing projective measurements. 
}\label{fig:syn_trunc_ans}
 \end{center}\end{figure*}
In Supp.~Fig.~\ref{fig:summaryplot}a we report a summary of the results for excited-state searches in the H$_2$ system, using 3 different excitation operators $\hat{E}_{p_j}$, and measuring projectors onto all the 9 non-degenerate excited subspaces that characterize the corresponding Hamiltonian. The set of $\hat{E}_{p_j}$ is expected to provide access to a portion of the spectrum with most of its support lying onto low--energy excited states, and a limited support onto higher energy excited states.
We notice how the optimization leads to a final fidelity $F$ with one of the excited subspace(s) exceeding $99 \%+$ in all cases (colour coded). 
In particular, $E_{p_3}$ produces an initial guess that is a superposition of eigenstates belonging to the $1^{\text{st}}$ and $4^{\text{th}}$ excited subspaces. This is a highly non-trivial case: because of the initial standard deviation adopted in the algorithm, the initial swarm of particles spans fidelities that are as low as $0.5$ with either subspace. Performing IPEA directly with this initial guess would perform poorly.
As emphasized in SM \ref{SM:foldedspectrum} adopting a Folded Spectrum method would also be of little help in this case, if insufficient information about the energy spectrum is available.\\
WAVES instead refines the initial guess detecting 87$\%$ of the times the $4^{\text{th}}$ subspace. This preferential collapse approximately matches what expected from the initial guess $\ket{\Psi_0}$, that has $0.78$ overlap with the $4^{\text{th}}$ subspace.

Equivalently, in Supp.~Fig.~\ref{fig:summaryplot}c are summarised the results for the H$_3^+$ molecule. The bigger dimensionality of this system leads to 23 non-degenerate excited subspaces. Among the cases investigated, we report only 3 different excitation operators, with the intent to reproduce a situation similar to what observed for H$_2$, and are thus relabelled as '\textit{Exc. 1-3}'. 
The optimization leads to a final fidelity $F$ with one of the excited subspace(s) exceeding $99 \%+$, in each case where the initial guess accurately targets a single excited subspace.  For excitation operator $E_{p_3}$, the preferential collapse of the particle swarm into either of the subspaces $4, 7$ and $23$ follows well what expected from the initial overlaps of $\ket{\Psi_0}$. 
We remark how a successful variational search for this last case required the adoption of a more accurate ansatz (the UCCSD in Supp.~Table~\ref{tab:ansatze}). $E_{p_3}$ produces an initial guess that is a superposition of eigenstates belonging to 7 different excited subspaces, yet predominantly targeting the $4^{\text{th}}$ excited-subspace (that has the highest overlap with the initial guess $\ket{\Psi_0}$). In order to access this whole portion of the Hilbert space and achieve a successful convergence, the simple PH ansatz would not suffice - see Supp.~Fig.~\ref{fig:unrestfixh3}. As shown in SM~\ref{sec:purityjust}, it is certainly possible for WAVES to provide a sanity check for the final eigenstate estimate obtained. In the same figure we outline how this may be performed in this case, making use of the eigenstate witness $\mathcal{S}$ 
and increasing the evolution time $t$ (for simplicity, we assume to operate always within a range where the Trotterization error is expected to be negligible). If no hardware decoherence issues arise for such longer evolution time, consistent degradation in the final $\mathcal{S}$ obtained at the end of a single run indicate a failure: either the ansatz has to be enriched, either the initial guess perturbed, for WAVES to achieve a successful convergence to a single $\ket{\Psi_{\tau}}$ or $\mathbb{E}_{\tau}$.\\
We will not report for brevity further analyses for the H$_3$ and H$_4$ systems, being the behaviour of WAVES in this case equivalent to findings already discussed for smaller systems.
\subsection{Synthetically reduced ans{\"a}tze}
\label{sec:redoscheme}
\begin{algorithm}[ht]
	\caption{Protocol for the synthetic truncation of the ansatz}
	\label{pseudocode:trunc}
	\KwData{Parametrized ansatz $\hat{A}(\vec{\theta}) \equiv \sum_{i \in I} (\theta_i \hat{A}_i )$, excitation operator $\hat{E}_{p_{\tau}}$ and targeted excited eigenstate $\ket{\Psi_{\tau}}$, threshold $\tilde{F}_0$ for the initial guess }
	obtain from WAVES optimal parameters for the ground state $\hat{A}_g = \hat{A}(\vec{\theta}_g)$, minimizing $\mathcal{F}_{\text{obj}} = -a\mathcal{P} + b\mathcal{E}$ \;
	calculate $\ket{\Psi_0} = \hat{E}_{p_{\tau}} \hat{A}_g \ket{\Phi}$ \;
	initialize $I' = I$ \;
	\While{ $| \langle \Psi_0 | \Psi_{\tau} \rangle |^2 \geq  \tilde{F}_0 $} {
		find $\iota$: $\theta_{g_{\iota}} = \min_i ( \{ | \theta_{g_i} | \})  $ \;
		$ I' \gets I' - \iota $ \;
		$ \hat{A_g'} = \sum_{i \in I} (\theta_{g_i} \hat{A}_i )$ \;
		calculate $\ket{\Psi_0} = \hat{E}_{p_{\tau}} \hat{A'_g} \ket{\Phi}$ \;
	}
\KwResult{truncated ansatz $ \hat{A'} (\vec{\theta}) \equiv \sum_{i \in I} (\theta_i \hat{A}_i )$, with $\dim(I') < \dim(I)$}
\end{algorithm}
One of the arguments often raised against variational methods is that they strongly rely on the accuracy of the ansatz chosen for the variational state preparation.
When considering more correlated and complex systems the accuracy of such ansatz is likely to be heavily limited. 
For this reason, trying to characterise variational methods in the presence of inaccurate or incomplete ans{\"a}tze is of fundamental importance for the assessment of their realistic performances. 

In the previous SM section we presented and discussed simulations mostly obtained with a standard PH ansatz, expected to achieve a reasonable accuracy in describing weakly correlated systems. Here we compare such results with a truncated ansatz $\hat{A'}$, obtained after a gradual removal of those terms in the PH ansatz, corresponding to the smallest values $|{\theta_{g}}|$. A more detailed breakdown of the procedure is reported in Algorithm~\ref{pseudocode:trunc}. 

$\hat{A'}$ thus simulates those cases where only limited knowledge on the physical system is available. 
The results are summarised for both ground and excited state searches in Supp.~Fig.~\ref{fig:syn_trunc_ans} a\&b, respectively. 
To stress the resilience of WAVES to a poor initial guess for the parameters, all the simulations in Supp.~Fig.~\ref{fig:syn_trunc_ans} were performed assuming no \textit{a-priori} knowledge of optimal $\vec\theta$, i.e. the initial guess $\hat{A}(\vec \theta_0) |\Phi\rangle \equiv |\Phi\rangle$.
For all the cases investigated, the variational search in WAVES is able to consistently identify both the ground and targeted excited-states. Fidelities of approximately $95\%$  are achieved, even for cases where the initial guess retained only small overlap with the targeted excited-state and the incomplete ansatz reduces the accuracy in the state-preparation. 
\begin{figure}[!ht]
\centering
\includegraphics[width=0.9\linewidth]{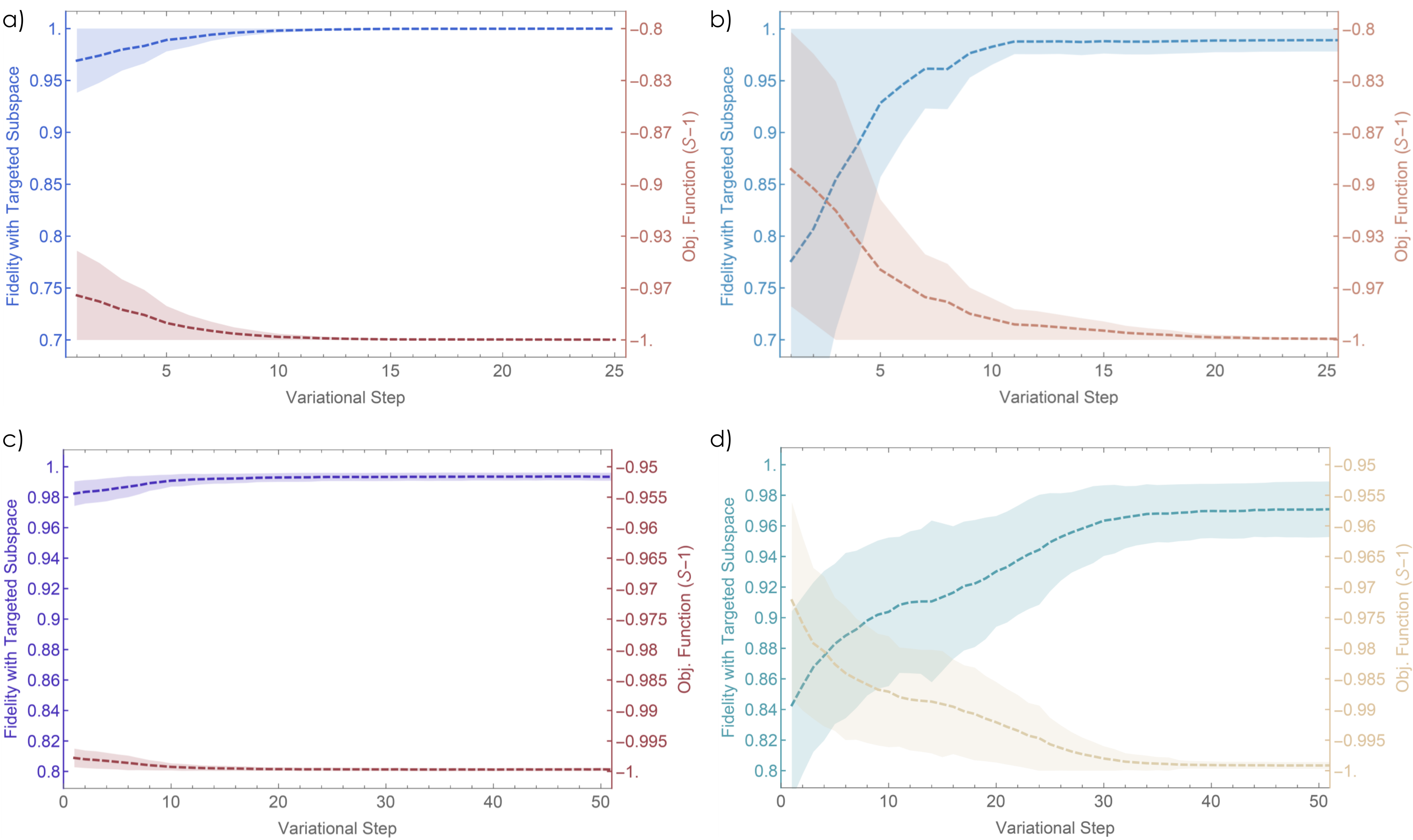}
\caption{Convergence of the WAVES algorithm to a subspace of excited states  in the case of the H$_2$ STO-3g (\textbf{a-b}, upper row) and H$_3^+$ STO-3g (\textbf{c-d}, lower row). On the left (a\&c) we report simulations averaged over 100 runs of the algorithm, adopting the standard PH ansatz. On the right (b\&d, emphasized by fainter colours), the algorithm is run with identical parameters, but adopting a truncated version of the same ansatz (refer to Algorithm~\ref{pseudocode:trunc}). 
Respectively $2$ and $6$ operators were truncated in cases b\&d.
}
\label{fig:pruningH}
\end{figure}
Details of the variational search for excited states in the H$_2$ and H$_3^+$ molecules are reported in Supp.~Fig.~\ref{fig:pruningH} a\&c (b\&d) respectively, comparing simulations with the standard (truncated) version of the PH ansatz (see also Supp.~Table~\ref{tab:ansatze}).
A decrease in the fidelity and uncertainty of the final estimate upon convergence are observed in the truncated versus the untruncated ansatz (compare also the colour-coded fidelities in Supp.~Fig.~\ref{fig:summaryplot}). Also, the restricted ansatz starts from a poorer initial guess and requires more variational steps to converge to the optimal value. Yet the algorithm displays a good resilience against such slightly inaccurate ans{\"a}tze, being still capable to detect the excited subspaces with high fidelity.
Again for the $H_2$ system, in the inset of Supp.~Fig.~\ref{fig:summaryplot}b, we observe how the reduced accuracy of this synthetic ansatz and overlap with the targeted eigenstate diminish the tendency of WAVES to collapse into the closest excited subspace. 

For the more complex case of the H$_3^+$ system, the difference in the final fidelities achieved, as well as in the frequency with which the swarm collapses into either one of the excited subspaces $\mathbb{E}_i$ are also summarised in Supp.~Fig.~\ref{fig:summaryplot}d.

A specific case is displayed in Supp.~Fig.~\ref{fig:pruningH} c\&d (corresponding to '\textit{Exc. 2}'). An accurate initial estimate $\ket{\Psi_0}$ leads to a very rapid convergence to fidelities with $\mathbb{E}_4$ exceeding 99\% for the PH ansatz in (c). In case (d), some of the crucial operators to explore the Hilbert space in proximity of the targeted subspace of excited have been evidently removed. This leads to a suddenly impoverished $\ket{\Psi_0}$, a consequently slower convergence, as well as a suboptimal fidelity with the targeted subspace $F \simeq 96 \%$. \\
\begin{figure}[!ht]
\centering
\includegraphics[width=0.9\linewidth]{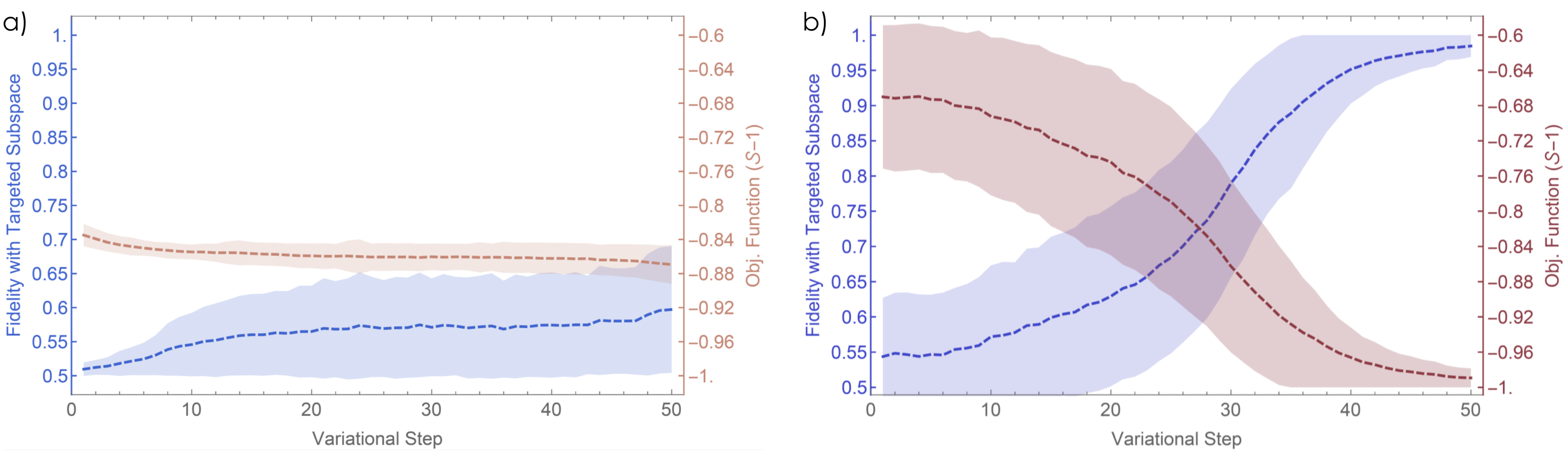}
\caption{
In \textbf{a}), we report the average results from 100 runs of the first part of the WAVES protocol adopting $\hat{E}_{p_3}$ for the PH ansatz in the H$_3^+$ system and $t=10$ in Methods Eq.~(9). The average entropy for the final estimate is $\mathcal{S}>$0.12. 
In \textbf{b}), an equivalent optimization adopts the unrestricted UCCSD ansatz from Supp.~Table~\ref{tab:ansatze}, but the same swarm-optimization parameters as in Algorithm~\ref{pseudocode}. 
A difference can be noticed in the 67.5\% confidence intervals for both $F$ and $\mathcal{S}$ at Step 0. This is due to attributing the same initial uncertainty to each variational parameter $\theta_i$, being $\dim \vec{\theta}$ in the unrestricted UCCSD ansatz about 8 times larger than in the standard PH case.
}
\label{fig:unrestfixh3}
\end{figure}

\subsection{Comparison with folded spectrum method}\label{SM:foldedspectrum}

We implemented a folded spectrum(FS)-like variational search~\cite{Wang_1994}, where we search for the state $|\Psi \rangle$ minimizing the energy \textit{folded} around the energy shift $\epsilon$, as in the modified eigenvalue Supp.~Eq.~(\ref{eq:fs}):
\begin{equation}
(\hat{\mathcal{H}} - \epsilon)^2 |\Psi \rangle =
(\mathcal{E} - \epsilon)^2 |\Psi \rangle
\label{eq:fs}
\end{equation}
which leads to the eigenstate closest in energy to the chosen $\epsilon$. Here we replace the usual conjugate gradient optimization method adopted in FS ~\cite{Tackett:2002} with a variational search for the state adopting the same particle swarm method in Algorithm \ref{pseudocode}, so to have a fair comparison ground with equivalent simulated runs of WAVES.\\
In this paragraph we report the results of applying FS to find excited states when applying the excitation operator $\hat{E}_{p_3}$ to the reference state of the H$_2$ systems, adopting a PH ansatz. We recall that $\hat{E}_{p_3}$ provides an initial guess state with good overlaps on both the 1$^{st}$ and 4$^{th}$ excited subspaces: $F_2 \simeq 0.39$ and $F_4 \simeq 0.61$ respectively. For a WAVES variational search, as shown in Supp. Fig. \ref{fig:summaryplot}, a majority of times the final state collapses on the 4$^{th}$ excited subspace, and in the remaining on the 1$^{st}$, always with high fidelity. This two-fold convergence averaged for 100 runs is reported in blue in Supp. Fig. \ref{fig:folded}.
\begin{figure}
\centering
\includegraphics[width=0.5\linewidth]{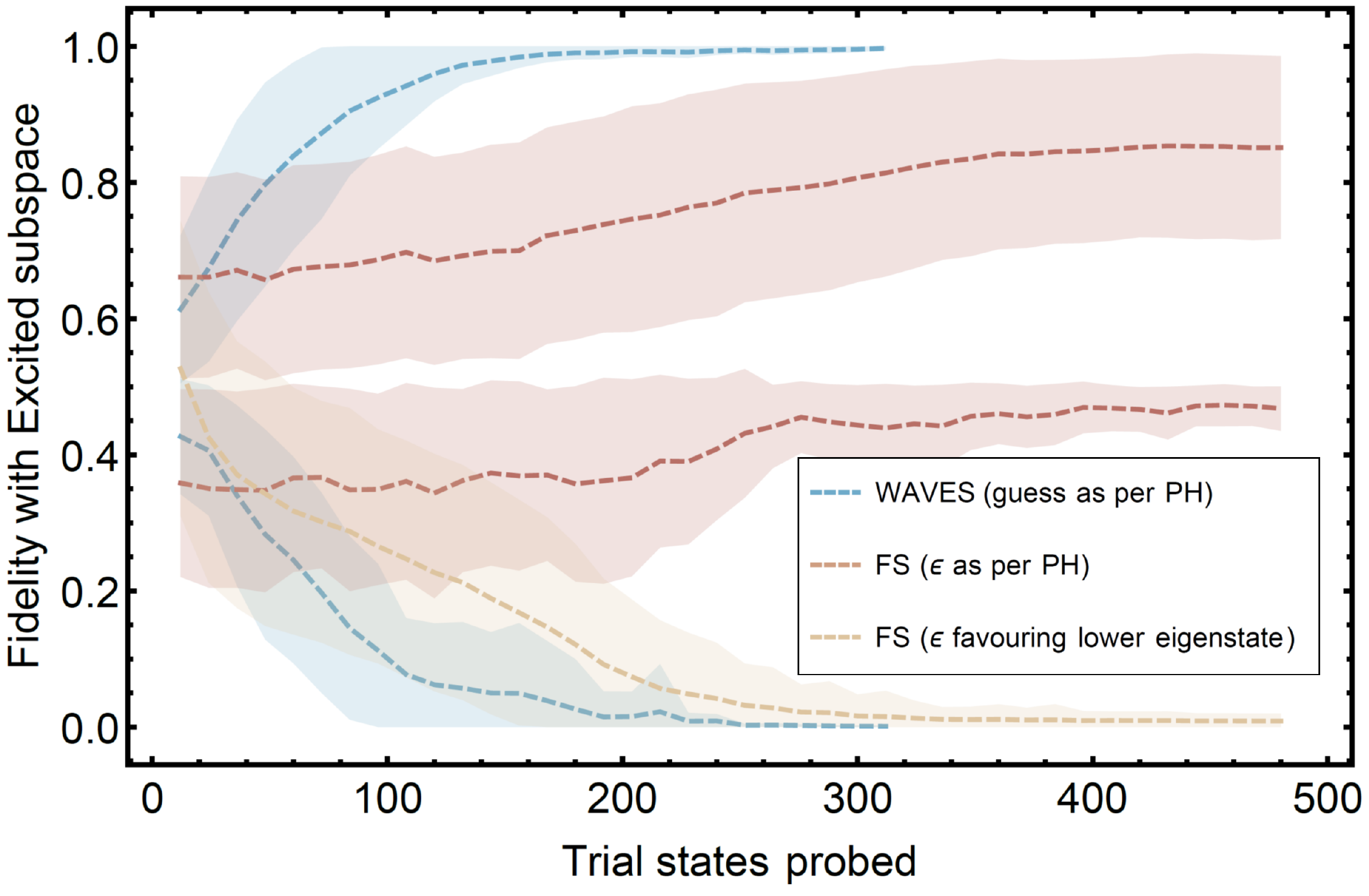}
\caption{
Behaviour comparison of the first part of WAVES, and an equivalent implementation of the Folded Spectrum method, when applied to the initial guess provided by the $\hat{E}_{p_3}$ excitation operator for the H$_2$ system. FS simulations were run adopting two different values of the energy shift $\epsilon$, colour-coded as in the legend. The fidelity reported is with the 4$^{th}$ excited subspace. We remark how all cases for both WAVES and FS implementation that converge to states with $F_4 \simeq 0$ have $F_1 > 0.99$ and can be considered successful searches. In all simulations, binomial noise was included (see SM \ref{sec:numsimH}).}
\label{fig:folded}
\end{figure}
In order to implement FS, a value for the energy shift $\epsilon$ is required. A natural choice is to employ the energy of the initial guess state provided by $\hat{E}_{p_3}$ applied to the PH parametrized ground-state: $|\Psi_3\rangle \equiv \hat{E}_{p_3} \hat{A}(\vec \theta_g) |\Phi\rangle $. Now, the energy of this trial state lies in between the two eigenvalues, corresponding to the excited subspaces with which we have consistent initial overlap. This makes them close in the folded spectrum and may become accidentally degenerate because of the initial uncertainty assigned to the set of parameters $\sigma_{\theta_g}$. If it is chosen $\epsilon \simeq \langle|\Psi_3 |\hat{\mathcal{H}} | \Psi_3\rangle$ in this way, the FS search in a minority of cases converges with high fidelity to the 4$^{th}$ excited subspace, but most of the time converges to local minima in the folded spectrum. These correspond to superpositions of eigenstates belonging to the different excited subspaces involved. In conclusion, the final excited-states estimates are poor, though still resembling the two-fold behaviour exhibited by WAVES (the results are shown in Supp. Fig. \ref{fig:folded}, in red). These limitations in the adoption of FS are known from the literature, including the original work~\cite{Wang_1994,Tackett:2002}.

In order to check the correct performance of the Folded Spectrum, we artificially reduced the value of $\epsilon$ to be slightly bigger than the eigenvalue of the 1$^{st}$ excited state. In this case, the FS search collapses consistently on the first excited subspace all the times, with a fidelity equivalent to that obtained with WAVES (see Supp. Fig. \ref{fig:folded}, in yellow). It is still possible to observe a slowdown in the convergence, requiring about twice the number of trial states with an equivalent choice of parameters. This poorer convergence rate is an expected effect, because of the squaring of the condition number, due in turn to the squaring of the relevant operator $(\mathcal{\hat{H}} - \epsilon)^2$. The resulting cost may be even worse in a quantum implementation, where the classical matrix multiplications to compute folded energies are replaced by dimension-efficient squaring in the terms of the Hamiltonian.  This tends to increase the cost of even a single energy evaluation by a polynomial amount on a quantum device.

This comparison highlights how a variational search implemented according to a FS method requires a careful choice in the energy shifts in order to produce reliable results. Obviously, whenever the energy spectrum is not known in advance, scanning different energies may be required to test the results obtained after a single run. For the same study case, WAVES variational search offers a solid alternative, successfully targeting excited eigenstates starting from poor initial guesses, with no a-priori knowledge of energy differences among the eigenstates.

\printbibliography

\end{document}